\documentclass[%
  reprint,%
  aps,%
  pra,%
  superscriptaddress,%
  longbibliography,%
  nofootinbib,%
  twocolumn
]{revtex4-1}

\usepackage{amsmath,amssymb,amsthm,amsfonts}
\usepackage{mathrsfs}
\usepackage{dsfont}
\usepackage{mathtools}
\usepackage{bm}
\usepackage{braket}
\usepackage{multirow}
\usepackage{array}
\usepackage{makecell}
\usepackage{stmaryrd}
\usepackage{amssymb}
\usepackage{comment}
\usepackage[normalem]{ulem}

\usepackage{graphicx}
\usepackage{color,xcolor}
\usepackage{titletoc} 

\usepackage{hyperref}
\hypersetup{
  colorlinks=true,
  linkcolor=blue,
  citecolor=blue,
  urlcolor=blue
}

\usepackage{qcircuit}
\usepackage{algorithm}
\usepackage{algpseudocode}

\usepackage{enumerate}

\newcommand{\Tr}{\mathrm{Tr}}

\newcommand{\T}{\mathcal{T}}

\newcommand{\rank}{\operatorname{rank}}

\newcommand{\comm}[2]{\left[#1,#2\right]}
\newcommand{\acomm}[2]{\left\{#1,#2\right\}}
\newcommand{\expect}[1]{\langle #1 \rangle}

\newtheorem{theorem}{Theorem}
\newtheorem{observation}{Observation}
\newtheorem{proposition}{Proposition}

\newtheorem{lemma}{Lemma}
\newtheorem{remark}{Remark}

\newtheorem{corollary}{Corollary}

\usepackage{varioref}
\labelformat{equation}{#1}

\labelformat{section}{#1}

\labelformat{figure}{#1}

\labelformat{proposition}{#1}

\labelformat{lemma}{#1}

\labelformat{theorem}{#1}

\labelformat{observation}{#1}

\labelformat{definition}{#1}

\labelformat{problem}{#1}

\labelformat{algorithm}{#1}

\labelformat{corollary}{#1}

\labelformat{statement}{#1}

\newcommand{\sun}[1]{\textcolor{black}{#1}}

\newcommand{\agg}[1]{\textcolor{purple}{#1}}

\begin{document}

\title{Quantum probe advantage in learning many-body systems}

\author{Wenzheng Dong}
\email{wenzheng.dong.quantum@gmail.com}
\affiliation{School of Physical and Chemical Sciences, Queen Mary University of London, London E1 4NS, United Kingdom}

\author{Andrew G. Green}
\affiliation{London Centre for Nanotechnology, University College London, Gordon St., London, WC1H 0AH, United Kingdom}

\author{Vlatko Vedral}
\affiliation{Clarendon Laboratory, University of Oxford, Parks Road, Oxford, OX1 3PU, United Kingdom}

\author{Jinzhao Sun}
\email{jinzhao.sun.phys@gmail.com}
\affiliation{School of Physical and Chemical Sciences, Queen Mary University of London, London E1 4NS, United Kingdom}

\begin{abstract}



Which properties of a quantum many-body system are operationally accessible is a central question underlying spectroscopy, thermodynamics, and quantum information science. Conventional response theory answers this question within a system-only paradigm: one perturbs and measures the matter itself, obtaining susceptibility built from causally ordered nested commutators. Here we show that coherently controlled quantum probes, when measured at the end, define a strictly larger operational learning framework beyond that accessible from response theory. We establish this through a quantum-circuit description that unifies spectroscopy, probe microscopy, and probe-based quantum technologies within a common operational framework, from which we develop quantum protocols for learning many-body properties from probe readout only. This advantage arises because the reduced dynamics of quantum probes generically encode anti-commutator and mixed-order correlators of the target; therefore, measurements on the probe provide access to fluctuations, non-equilibrium structure, and entanglement entropy that are in general not accessible through response functions or a single probe alone. Moreover, we demonstrate that entangled probes can access many-body properties such as von Neumann entropy. We prove that the required probe resources scale with the complexity of the target correlations rather than with the size of the many-body system. Quantum probes are therefore not merely more sensitive sensors but provide a new way to learn many-body properties distinct from those of tomography or quantum simulation.

\end{abstract}

\maketitle



\vspace{10pt}


\section[0]{Introduction}

Which properties of quantum matter are operationally accessible?
For decades, the answer has been shaped by response theory~\cite{kubo1957statistical}. In spectroscopic experiments~\cite{braicovich2010magnetic,klepp2014fundamental} one perturbs the system, e.g. with a beam of neutrons, and reads out its ensuing dynamics, thereby reconstructing susceptibilities from how the system reacts to external stimuli~\cite{boothroyd2020principles}. Conversely, macroscopic quantities such as magnetic susceptibility and conductivity~\cite{roman2025linear} are given – through the Kubo-Greenwood formulae~\cite{greenwood1958boltzmann} – as retarded correlators of local operators. This paradigm forms the cornerstone of probing many-body properties and understanding many-body physics~\cite{marconi2008fluctuation,roggero2019dynamic,baez2020dynamical}.

Quantum engineering brings this question back into focus from a new perspective. 
Across atomic, solid-state, and photonic settings, 
controllable ancillae have become central to applications in quantum technologies and probe microscopy,   including superconducting scanning tunnelling microscopy (STM) tips~\cite{binnig1982surface,kim2018toward}, nitrogen-vacancy  centres in diamond~\cite{balasubramanian2008nanoscale,boss2017quantum}, cavity modes~\cite{dorfman2016nonlinear,xia2023entanglement}, \sun{or scanning-probe sensors~\cite{huxter2025multiplexed},} that couple coherently to a target and can themselves be measured with high fidelity, respectively. Indeed, nuclear magnetic resonance (NMR) uses the simplest of quantum probes~\cite{gershenfeld1997bulk,alvarez2011measuring} – a single spin – and will be used to illustrate our ideas below.  
Once such auxiliary degrees of freedom are treated as explicit quantum systems~\cite{rossini2007decoherence} rather than merely as sources of perturbation, the traditional response-theory picture is no longer the only operational framework available~\cite{wang2019characterization}. 
This shift raises a fundamental conceptual question that remains open: Do quantumness and entanglement of probes fundamentally enlarge what can be learned about matter~\cite{chiara2006scheme,irfan2021quantum}? And more importantly, how can entangled probes learn many-body systems? This is not straightforward; direct measurements on the system might appear more powerful since they provide direct access to its dynamics and can involve substantially more degrees, though in practice the number of conventional probes is often restricted.


Here we first show at the operator level that measurements on coherently controlled quantum probes fundamentally enlarge what can be learned about many-body systems. \sun{
A key ingredient of our framework is the quantum-circuit description shown in Fig.~\ref{fig:fig_1_demo}. This provides a unified operational model encompassing spectroscopy (and its quantum simulation)~\cite{roggero2019dynamic,kokcu2024linear,sun2025probing,bruschi2024quantum,lee2021simulation,fomichev2024simulating,chan2025algorithmic}, scanning-probe microscopy~\cite{huxter2023imaging,huxter2025multiplexed}, and probe-based quantum techniques~\cite{chiesa2019quantum,stenger2022simulating,gjonbalaj2026unified,wang2021classical,wang2019characterization}, and moreover, provides the protocol for learning many-body properties from the (entangled) probes rather than direct measurements on the system.
}
The advantage of quantum probes arises because the reduced dynamics of the probe-based approach encode correlator sectors beyond the causally ordered response functions of system-only measurements~\cite{dong2025efficient}. \sun{Specifically, quantum probe can learn system fluctuations beyond response theory by virtue of an enlarged operator space, which encompasses $2^{k-1}$ sectors of $k$-point correlators, including anti-commutators and their mixed orderings. In contrast, system-only approaches provide access only to commutators. Therefore, measurements on the probe provide access to fluctuations, non-equilibrium information}, and entanglement-related observables such as entropy that are not, in general, determined by response theory \sun{nor single-probe alone.}

The notion of quantum probe has broader implications and connections to established techniques. This includes proposals, such as entangled probes for entanglement extraction~\cite{chiara2006scheme} based on neutron interferometry~\cite {erdosi2013proving}, and experimental techniques such as NMR and superconducting STM tips using Cooper pairs to probe more entangled states than metallic tips~\cite{kim2018toward}. 
There has been substantial development in multimodal scanning probe spectroscopy and quantum sensing theory~\cite{gjonbalaj2026unified,wang2021classical,wang2019characterization} and experiments~\cite{huxter2023imaging,huxter2025multiplexed}. 
In quantum technology, a related example of this probe-based viewpoint arises through quantum noise spectroscopy in which controlled quantum probes can be used to infer environmental noise spectra~\cite{alvarez2011measuring,yoneda2023noise,yan2013rotating,sung2021multi} \sun{while here the objective is to extract many-body features. Recent developments have extended linear response theory to non-Hermitian perturbations~\cite{pan2020nonHermitian,geier2022non} though they still remain within the linear response paradigm.}
 


To establish the probe advantage in extracting many-body properties, we generalise to multiple entangled probes and show that entangled probes systematically enhance learnability. 
As a concrete application, we demonstrate that the entanglement entropy expanded to order $K$ can be learned by $\mathcal{O}(K)$ entangled probes, which is independent of system size. More generally, to access matter correlators up to order $K$, it suffices to employ $\mathcal{O}(K)$ entangled probes. This is different from nonlinear spectroscopy~\cite{cheung2025quantum,wang2019characterization}, such as two-dimensional coherent spectroscopy~\cite{wan2019resolving,nandkishore2021spectroscopic} and entangling matter with a single sensor~\cite{meinel2022quantum, pfender2019high,wang2019characterization,wang2021classical}, in both learnability and operational point of view, where the probes are not entangled. Our work, based on a quantum circuit model, provides an efficient protocol to access the many-body properties. Typical quantum algorithms, however, usually need full access to the system evolution (which could be regarded as a system-based approach). For example, computing higher-order correlators typically requires controlled evolution~\cite{roggero2019dynamic,kokcu2024linear} or imaginary-time evolution~\cite{wang2025computing}.    Our work thus establishes a fundamentally new route to learning many-body properties, distinct from tomography-based methods~\cite{huang2020predicting,miquel2002interpretation} or quantum algorithms~\cite{roggero2019dynamic,kokcu2024linear,sun2025probing,bruschi2024quantum,lee2021simulation,fomichev2024simulating,chan2025algorithmic}, which do not explicitly exploit the system's correlation structure.

\begin{figure*}[htbp]
  \centering
  \includegraphics[width=1.0\textwidth]{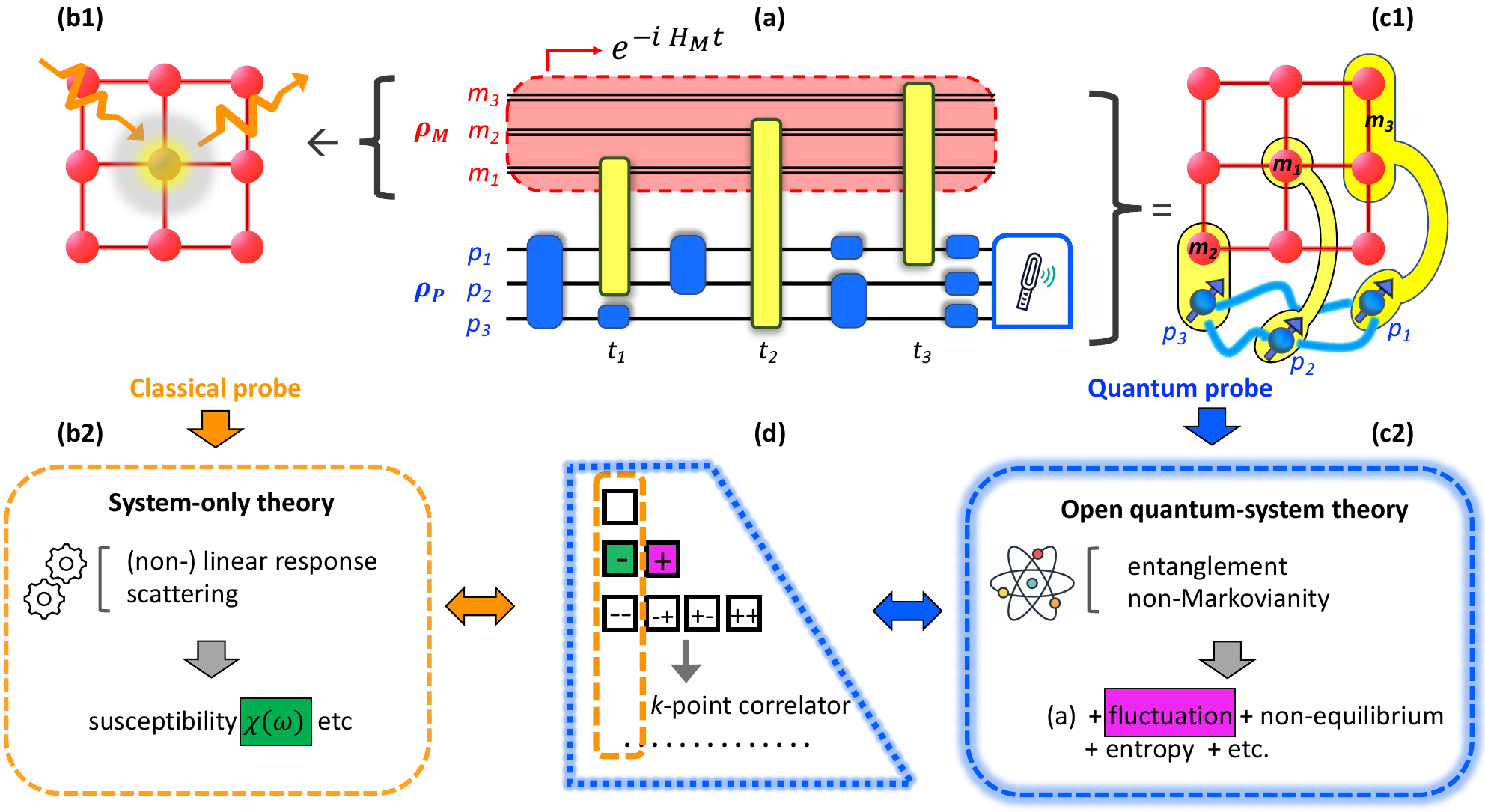}
  \caption{\textbf{Framework of quantum-probe learning: coherent ancillae turn spectroscopy into an open-system readout problem.}
  (a) Quantum-circuit representation of the probe protocol: a controllable probe register $\rho_P$ undergoes coherent control, couples through addressable channels $p$ to selected matter operators $M_{m(p)}$ in the target state $\rho_M$, and is finally measured so that information about matter is inferred from reduced probe observables.
  (b1)--(b2) Effective classical-probe regime: when the probe is treated as a prescribed perturbation rather than as a measured quantum subsystem, the circuit reduces to the target system alone; correspondingly, scattering leads to system-only response theory with susceptibilities such as $\chi(\omega)$.
  (c1)--(c2) Genuine quantum-probe regime, where coherent probe control (blue curves), probe--matter entanglement (yellow bulbs), and non-Markovian open-system dynamics transfer richer information from matter to probe. The example in (c1) uses $(m(p_1),m(p_2),m(p_3))=(m_3,m_1,m_2)$, with $m_3$ denoting an extended two-site subsystem.
  (d) Corresponding learning spaces for a $k$-point correlator: the orange dashed strip is the single nested-commutator sector available to system-only response, whereas the surrounding blue dashed region denotes the broader family of mixed commutator/anti-commutator sectors opened by quantum probes. For $k=2$, these reduce to the response and fluctuation sectors, respectively.}
  \label{fig:fig_1_demo}
\end{figure*}
\section[1]{Quantum versus Classical Probes}
\label{sec:probe_framework}

Let us start with what we mean by an effective classical-probe regime.
The term does not simply imply that the physical probe itself, such as a beam of neutrons or photons, is literally classical.
Rather, it refers to an operational regime in which the probe's own quantum state is not used as the explicit information-bearing readout register.
In that effective description, the probe acts only through the perturbation it applies and the signal it helps induce, while the learned quantities are encoded in observables of the target system itself.
\sun{In this sense, a classical probe effectively acts as a c-number source, $\lambda(t)M(t)\equiv (\lambda(t)I_P)\otimes M(t)$ in a formal enlarged space, such as that in neutron scattering, whereas a quantum probe is represented by an operator-valued coupling $h_P(t)\otimes M(t)$ (see Methods).
}

The classical-probe scenario is precisely the one described by response theory.
The target is driven by externally prescribed perturbations and subsequently measured, so the effective circuit description contains only the matter degrees of freedom as explicit quantum wires; schematically, this is the regime illustrated in Fig.~\ref{fig:fig_1_demo}(b1)--(b2).
The resulting observables are susceptibilities and higher-order response functions built from causally ordered commutator structures.
What this framework does not exploit is the probe as an independent quantum memory of the interaction history.

The quantum-probe regime is operationally different.
Here, the probe remains an explicitly modelled, coherently controlled quantum system throughout the protocol: it is initialised, driven, coupled to selected matter operators, allowed to become entangled with the target, and finally measured.
In this setting, the reduced final state of the probe itself becomes the readout channel for matter information, as depicted schematically in Fig.~\ref{fig:fig_1_demo}(a) and Fig.~\ref{fig:fig_1_demo}(c1)--(c2).

The essential resource in this regime is probe-matter entanglement and, more generally, {the controlled history by which that entanglement is created and accessed.}
Because the probe is read out only after a designed sequence of coherent interactions, its final observables carry a record of how the target has imprinted itself onto the probe over the full interaction history, with non-Markovian information flowing into the probe.

Viewed in this way, the operational shift is not that one must build a probe register comparable in size to the many-body system.
Rather, one uses a small, well-controlled quantum register and redesigns the interaction schedule, coherent control, and final readout so that the probe functions as a compact quantum circuit for learning matter (see Fig.~\ref{fig:fig_1_demo}(a)).
{This is the key message of the framework: for platforms that already possess controllable ancillae, the decisive step is often not upgrading the hardware (i.e. adding more probes), but a different orchestration of the same probe--matter dynamics.}
{
From this starting point we derive a unified learning interface that separates what the laboratory can \emph{dial} (probe controls, couplings, etc.) from what the matter can \emph{encode} (multi-time correlators of target matter).}


It is worth clarifying the distinction between quantum and classical probes. A classical probe still interact with a quantum system and thus require the quantum resources. However, the key is that it functions only as an external perturbation or measurement channel and thus the accessible information is constrained by the system dynamics alone. By contrast, this work exploits the probe's quantumness to enlarge the class of learnable quantities.

This shift may yield a broader notion of learnability.
Once the probe is treated as a measured open quantum subsystem rather than as a hidden perturbation source, the operationally accessible information channel need no longer be confined to the standard response sector.
The above insights of quantum probe advantage are algebraic: standard response theory is confined to fully retarded commutator correlators, whereas quantum probes naturally access the $2^{k-1}$ sectors of $k$-point correlators, including anti-commutators and mixed orderings.
The simplest nontrivial instance already appears for a single coherently controlled probe, whose final coherence separates response-like and fluctuation-like information \sun{as extensively studied in other contexts~\cite{geier2022non,wang2021classical,wang2019characterization,pfender2019high}}.
We turn to that minimal example below.


\section[2]{Power of A Single Quantum Probe}
\label{sec:probe_decoherence}

{To illustrate our framework, we first consider a probe consisting of a single quantum spin as a starting example. This is a reframing of the familiar schemes used in NMR, where independent nuclear spins form the probe, or muon spin resonance ($\mu$-SR) where implanted muons provide the probe spins. In this context it is well known that the probe spin dynamics is affected by its environment in two ways leading to energy relaxation ($T_1$) and dephasing ($T_2$). If the system being studied is in thermal equilibrium, these are related by a fluctuation-dissipation relation. In general, however, they provide distinct information about the system. }

This is formalised in the two canonical correlators of a many-body interaction-pictured observable 
$M(t)$: the commutator correlator $\mathcal{M}^-(t,t')=\langle [M(t),M(t')] \rangle /2$ and the anti-commutator correlator $\mathcal{M}^+(t,t')=\langle \{M(t),M(t')\} \rangle /2$. 
For Hermitian $M(t)$, $\mathcal{M}^-$ and $\mathcal{M}^+$ are imaginary and real respectively. 
They represent two physically distinct sectors of many-body dynamics: the commutator sector $\mathcal{M}^-$ generates causal response (the energy decay in our single spin example) and {can be detected using standard linear response theory, see Fig.~\ref{fig:fig_1_demo}(b2)}; 
by contrast, $\mathcal{M}^+$ characterises fluctuations of $M$ (dephasing of our single spin) and is not captured by conventional response theory. 
The common bridge between the two sectors is the fluctuation–dissipation theorem, which relates $\mathcal{M}^+$ and $\mathcal{M}^-$ only under the special assumption that the system is in thermal equilibrium.  
Away from equilibrium, the two objects encode independent information about the many-body system. \sun{Note that in another context, a single quantum probe has been found to extract commutators and anti-commutators in quantum sensing~\cite{wang2021classical,wang2019characterization,pfender2019high} or non-Hermitian perturbation~\cite{geier2022non}. However, our central objective is to learn the matter dynamics rather than probe dynamics. To that end, we design an efficient algorithm for learning many-body properties thanks to the unified circuit-model framework developed in \autoref{fig:fig_1_demo}.
}

A quantum probe provides a direct operational channel to this otherwise hidden fluctuation sector, {as indicated in Fig.~\ref{fig:fig_1_demo}(c2) for its Fourier-transformed form.} 
To see this, consider the simplest setting: a single qubit probe $P$ longitudinally coupled to a matter observable $M(t)$ through an interaction $H_I(t)=\frac12 [\cos\theta(t)\sigma_x+\sin\theta(t)\sigma_z]\otimes M(t)$, where $\theta(t)$ is a controllable modulation generated by probe control and the initial state is factorized as $\ket{+}\bra{+}_P\otimes\rho_M$. 
The probe coherence at time $T$ is captured by the expected off-diagonal element $\mathbb{E}[\rho_{01}(T)]= e^{\Gamma(T) + i\Phi(T)}$.
Note that for simplicity, we assume the matter's statistics is a zero-mean Gaussian distribution. With some derivation, one obtains the expression of the decay exponent $\Gamma$ and the coherent phase shift $i\Phi$ in terms of the two correlators $\mathcal{M}^+$ and $\mathcal{M}^-$:
\begin{equation}
\begin{aligned}
    \Gamma(T) &=-\ln2-\int_0^T\!dt\!\int_0^t\!dt'\sin\theta(t)\sin\theta(t')\mathcal{M}^+(t,t'),
\\
\Phi(T) &=-i\!\int_0^T\!dt\!\int_0^{t}\!dt'\sin(\theta(t)-\theta(t'))\mathcal{M}^-(t,t').
\end{aligned}
\end{equation}
These two terms have a transparent interpretation in the language of qubit dynamics. 
The quantity $\Gamma$ governs the irreversible decay of probe coherence: it is the genuinely open-system part of the probe evolution. 
Remarkably, it depends exclusively on the fluctuation kernel $\mathcal{M}^+$. 
The probe therefore converts microscopic fluctuations of the many-body observable $M$ into a directly measurable decoherence signal. 
In contrast, the phase $i\Phi$ represents a coherent unitary phase shift of the probe and depends only on the commutator correlator $\mathcal{M}^-$. 
This is precisely the retarded response sector accessed in conventional susceptibility measurements. 

The probe–matter interaction therefore separates two complementary channels of information flow. 
The unitary part of the probe evolution, encoded in $i\Phi$, mirrors the causal response of the matter system and corresponds to the same commutator structure that underlies response theory. 
The non-unitary part, encoded in $\Gamma$, reflects the fluctuations of the matter system that have flowed into the probe through probe–matter entanglement and appeared as irreversible decoherence in the reduced probe dynamics. 
Since the joint evolution of $P$ and $M$ is unitary, this decoherence does not represent a loss of information at the global level; rather, it is a transfer of fluctuation information from $M$ into correlations with the probe. 
From this perspective, the probe acts as a detector that translates the fluctuation sector of many-body physics into experimentally accessible signals.

The key message is therefore not that response theory is ineffective — the matter perturbed by classical probe is described by a unitary theory, and it is structurally incomplete for open-system learning. 
A quantum probe, by contrast, inevitably experiences open dynamics once it becomes entangled with the system, and this open dynamics exposes the anti-commutator sector $\mathcal{M}^+$. 
Even the minimal dephasing experiment thus reveals a fundamental extension of the operational information accessible in many-body physics: probe decoherence transforms intrinsic fluctuations of a quantum system into a measurable signal.


We  formulate the following statement of probe advantage in this minimal setting:
\begin{observation}[Quantum probe: learning fluctuation beyond response theory] 
Given the setup described above, a quantum probe can learn fluctuation beyond response theory without relying on the fluctuation-dissipation theorem.
  \label{theo:decoherence}
\end{observation}
The Observation~\ref{theo:decoherence} is a natural consequence of Theorem~\ref{theo:probe_adv} in the subsequent Section.

\begin{figure*}[tbp]
  \centering
  \includegraphics[width=0.8\textwidth]{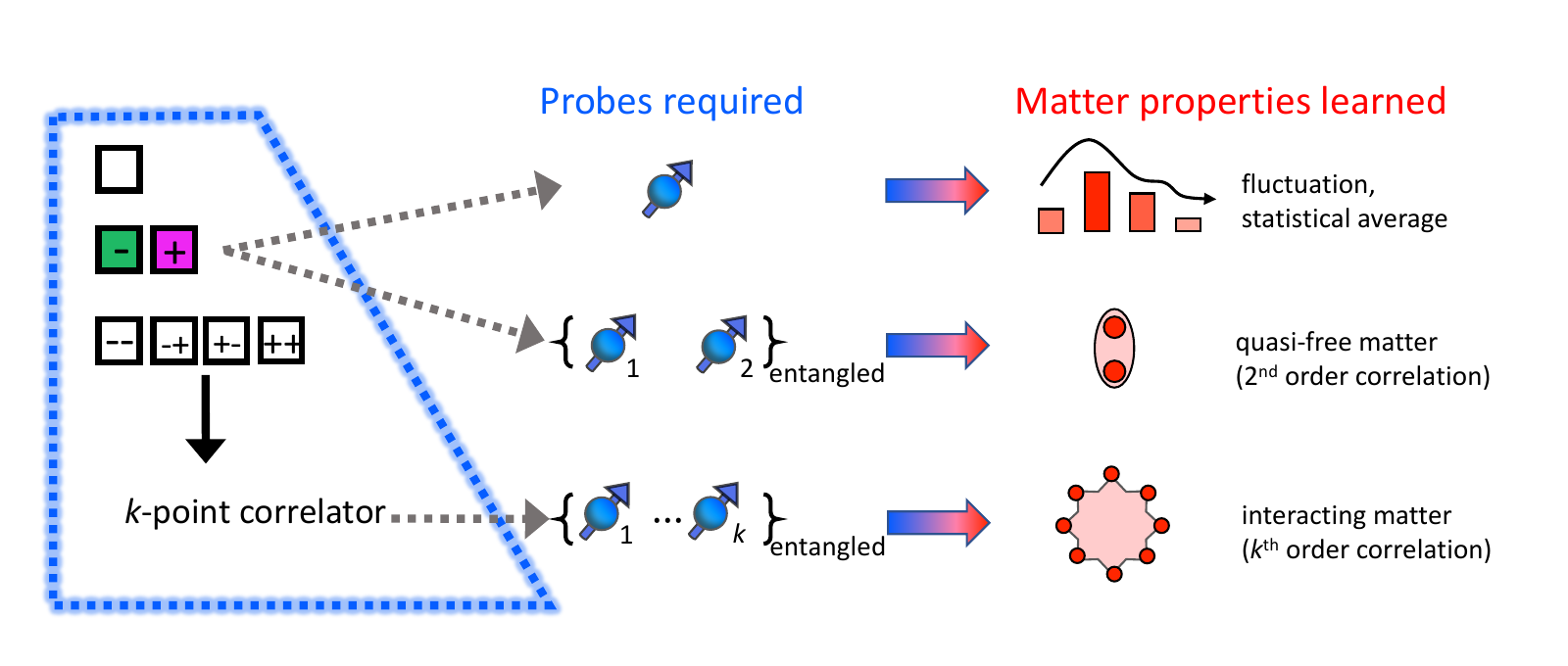}
  \caption{\textbf{Probe resources are set by correlator complexity, not system size.}
  The left wedge represents the enlarged $k$-point correlator landscape opened by probe-based learning, with mixed commutator/anti-commutator sectors beyond the single retarded response branch.
  The middle column indicates the minimal probe resources needed to access different levels of many-body structure: a single probe already resolves fluctuations and statistical averages; two entangled probes suffice for quasi-free matter, where second-order correlators determine the relevant physics; and $k$ entangled probes access genuinely interacting matter through $k$th-order correlator structure.
  The right column summarizes the corresponding learned properties, emphasizing that the probe overhead scales with the order of the target correlations rather than with the full system size.}
  \label{fig:fig_2_probe_circuit}
\end{figure*}



\section[3]{General Algebraic Structure and Protocol of Quantum Probe-Based Learning}
\label{sec:probe_landscape}

With the dephasing example in Sec.~\ref{sec:probe_decoherence} as a warm-up, we now step back and ask a fundamental question: {what is the universal object a quantum probe learns about matter, and how does this differ from system-only response?}
This section sets up a compact algebraic landscape that will be used repeatedly later, while keeping the discussion operational:
every statement below is phrased in terms of {what a laboratory can dial} (probe control and coupling) and {what it can read out} (final probe observables).


\paragraph{Setting.}
We consider a controlled probe register $P$ (one or several qubits/modes) coupled to a many-body system $M$.
The total Hamiltonian is $H(t)=H_P(t)+H_M+H_{PM}(t)$, where $H_P(t)$ is a user-controlled probe drive and $H_M$ is the intrinsic matter Hamiltonian.
The only structural assumption is that the $P-M$ interaction  is a sum of single-probe-channel operators.
Each probe channel is assigned to a chosen matter operator by a fixed channel-to-matter assignment $p\mapsto m(p)$:
\begin{equation}
\label{eq:general_coupling}
  H_{PM}(t)=\sum_{p\in P} h_p(t)\otimes M_{m(p)},
\end{equation}
In Eq.~\eqref{eq:general_coupling}, $P$ labels the addressable probe--matter coupling channels available in the probe register.
In the main-text protocols we use the direct-addressing implementation, where each channel $p\in P$ is carried by one coherently controlled probe degree of freedom.
Thus the number of channels and the number of probes coincide in each run, which is the convention used in the resource statements below.
The operator $h_p(t)$ acts on the probe degree of freedom carrying channel $p$, e.g., $h_p(t) = f_p(t)Z_p$ with $f_p(t)$ being a time-dependent modulation, while $M_{m(p)}$ is the matter operator addressed by that channel.
Here $p$ is a probe-channel label and $m(p)$ is the assigned matter label; they need not coincide.
The matter operator $M_{m(p)}$ may be supported on a site, a bond, or an extended region.
For example, Fig.~\ref{fig:fig_1_demo}(a,c1) illustrates three probe channels with
$(m(p_1),m(p_2),m(p_3))=(m_3,m_1,m_2)$.
Here $m_3$ is an extended two-site subsystem, showing that $m(p)$ can label a site, bond, or larger region; the operator in the $j$th probe slot is therefore $M_{m(p_j)}$, not necessarily $M_{m_j}$.
In the single-probe dephasing example of Sec.~\ref{sec:probe_decoherence}, this bookkeeping is trivial: there is only one addressed matter operator, so the channel label carries no nontrivial subsystem assignment.
More general interfaces, in which time-windowed channel reassignment or probe reuse can make the required number of physical probes smaller than the number of addressed channels, are discussed in the Supplementary Material.

\paragraph{From probe readout to matter correlators.}
All experimentally accessible information about $M$ enters through probe observables measured at a final time $T$,
$\langle O_P(T)\rangle_{\rho_P}
=\Tr_{PM}\!\big[(O_P\otimes \mathbb{I}_M)\,U_I(T)\,(\rho_P\otimes\rho_M)\,U_I^\dagger(T)\big],$
where $U_I$ is the propagator in the interaction picture with respect to $H_P(t)+H_M$.
Because probe--matter entanglement generically develops, the reduced probe state $\rho_P(T)$ is an open-system object whose Dyson expansion naturally produces \emph{nested} operator structures on $M$.

A compact way to organize those structures is to keep track of whether each nesting step is a commutator or an anti-commutator\cite{dong2025efficient}.
We use the ``braketor'' shorthand
$\llbracket X,Y\rrbracket_{\mu} := XY+(-1)^{\mu}YX$, with $\mu\in\{0,1\}$, so that $\mu=1$ gives a commutator and $\mu=0$ gives an anti-commutator.
With this notation, the probe expectation admits the universal expansion~\cite{dong2025efficient} (see {Methods} for the derivation and precise definitions of the control tensor):
\widetext
\begin{equation}
  \begin{aligned}
  \langle O_P(T)\rangle_{\rho_P}
&=\sum_{k\geq 0, ~\vec{p}, ~\vec{\mu}}(-i)^k
\int_0^T d_>\vec{t}\;
\mathcal{F}^{(k;\vec{\mu})}_{O_P;\vec{p};\vec{t}}~\big[H_P(\vec{t})\big] \times 
 \mathcal{M}^{\vec{\mu}}_{\vec{p}}(\vec{t}),   \\
 \mathcal{M}^{\vec{\mu}}_{\vec{p}}(t_1,\ldots,t_k)
&:=
\frac{1}{2^{k-1}}\,
\widehat{P}_{\vec{p}}\Big(\expect{
\llbracket\cdots\llbracket
M_{m(p_1)}(t_1),M_{m(p_2)}(t_2)\rrbracket_{\mu_1},
M_{m(p_3)}(t_3)\rrbracket_{\mu_2}\cdots,
M_{m(p_k)}(t_k)\rrbracket_{\mu_{k-1}} }_{\rho_M}
\Big).
  \end{aligned}
  \label{eq:dyson_commutator}
\end{equation}
\twocolumngrid
The meaning of each symbol is as follows: the sum is over all orders $k$ of the Dyson expansion, $d_>\vec{t}$ denotes the time-ordered integration region $t_1\ge t_2\ge\cdots\ge t_k$,
$\vec{p}=(p_1,\ldots,p_k)$ indexes the probe-channel slots generated by the expansion, and the actual matter operators in those slots are written explicitly as $M_{m(p_j)}(t_j)$.
The binary vector $\vec{\mu}=(\mu_1,\ldots,\mu_{k-1})$ specifies the braketor sector: $\mu_j=1$ inserts a commutator at the $j$th nesting step, while $\mu_j=0$ inserts an anti-commutator.
The operator $\widehat{P}_{\vec{p}}$ is retained because the nested-braketor basis compares different operator orderings while keeping the external probe-channel tuple in a fixed canonical order.
For $k=2$, this means $\mathcal{M}^{(\mu)}_{p_1,p_2}(t_1,t_2)=\frac12[\expect{M_{m(p_1)}(t_1)M_{m(p_2)}(t_2)}_{\rho_M}+(-1)^\mu\expect{M_{m(p_1)}(t_2)M_{m(p_2)}(t_1)}_{\rho_M}]$ with $\mu=0,1$, where the second term is the $\widehat{P}_{(p_1,p_2)}$-relabelled reversed ordering and the subscript remains the probe-channel pair, not a matter-index pair.
The functional $\mathcal{F}^{(k;\vec{\mu})}_{O_P;\vec{p};\vec{t}}[H_P(\vec{t})]$ is a {control tensor}:
it is fully determined by the probe Hilbert space, the chosen initial state $\rho_P$, the applied probe control $H_P(t)$, and the final observable $O_P$.
In contrast, all dependence on the unknown state and dynamics of $M$ is isolated in the correlators
$\mathcal{M}^{\vec{\mu}}_{\vec{p}}(\vec{t})$, whose subscript records the probe-channel history and whose entries contain the corresponding matter operators $M_{m(p_j)}$.
Equation~\eqref{eq:dyson_commutator} is therefore the universal interface between ``dial settings'' and ``matter information'': the learning task is to engineer $\mathcal F$ so that the desired correlator sectors become identifiable from measured probe data.
Although it is written as a time-ordered convolution over insertion times, it is not meant to be inverted as a brute-force high-dimensional kernel.
The role of the probe control tensor $\mathcal F$ is precisely to shape this kernel: in the windowed protocols developed in the Methods and Supplementary Material, short probe--matter interaction windows localize the integral near chosen times, reducing the readout to calibrated linear equations for selected correlator sectors.
Broader or always-on couplings lead to the same inverse-problem structure after deconvolution.

\paragraph{The algebraic distinction and a physical picture.}
Equation~\eqref{eq:dyson_commutator} already exposes the conceptual separation between {system-only response theory} and {quantum probe-based open-system learning}.
Because the probe is an external quantum degree of freedom that is \emph{coherently} driven and then traced down to a reduced observable, its readout decomposes into correlator sectors labelled by $\vec{\mu}$, i.e., by the choice of commutator ($\mu_j=1$) or anti-commutator ($\mu_j=0$) at each nesting level.
Nothing in the reduction enforces $\vec{\mu}=(0,\ldots,0)$.
By contrast, conventional response protocols perturb and measure $M$ itself, and causal ordering restricts the learned objects to the fully retarded edge: nested commutators only.
This contrast is already explicit at $k=2$: response isolates $\mathcal{M}^-$, while a probe can also resolve $\mathcal{M}^+$.

A useful way to remember \emph{why} the number of sectors grows is to think in time-ordered ``interaction windows'' $t_1\!>\!t_2\!>\!\cdots\!>\!t_k$.
The Dyson expansion forces the nesting to follow this time ordering, but at each new insertion, there is a binary choice: the additional matter operator can enter the reduced probe observable through a \emph{response-like} channel (a commutator with what came before) or a {fluctuation-like} channel (an anti-commutator).
Iterating this choice produces a branching, time-ordered tree with $2^{k-1}$ leaves at order $k$.
In this sense, a quantum probe can ``learn sequentially'' along the interaction history, combining response and fluctuation information in all algebraically distinct ways.
The Gaussian setting discussed in Sec.~2 is the simplest instance of this landscape, where only the lowest sectors are needed; beyond Gaussianity, the higher-order branches encode genuinely new multi-time structure.

We emphasize that the branching picture above is a mnemonic for the algebra rather than a complete physical taxonomy.
The exact content of Eq.~\eqref{eq:dyson_commutator} is that probe readout resolves a family of nested commutator/anti-commutator correlators, while system-only response is confined to the fully retarded sector.
Beyond the lowest orders, labels such as ``response-like'' and ``fluctuation-like'' should therefore be understood only heuristically.
Out of equilibrium the different sectors are generally independent; in equilibrium they may be linked by KMS-type constraints, with the fluctuation--dissipation relation at second order as the simplest familiar case.
For our purposes, Eq.~\eqref{eq:dyson_commutator} should be read as an operator-level map of what probe measurements can access.

This discussion can be summarized in the following structural counting statement.
\begin{theorem}[Quantum probe advantage]
  Regardless of the probe size, even for a one-qubit probe, a quantum probe can learn $2^{k-1}$ distinct correlator sectors while a system-only (classical) response protocol can only learn one.
\label{theo:probe_adv}
\end{theorem}

Theorem~\ref{theo:probe_adv} concerns which operator sectors are accessible in principle, not yet how they are reconstructed in practice.
Its simplest realization is the single-probe dephasing setting of Sec.~\ref{sec:probe_decoherence}.
For $k=2$, the two available sectors are precisely the commutator $\mathcal{M}^-$ and anti-commutator $\mathcal{M}^+$, so Observation~\ref{theo:decoherence} is the lowest-order instance of the theorem.

\paragraph{From a structural statement to a learning protocol.}
Theorem~\ref{theo:probe_adv} is a structural statement about accessible operator orderings, not a promise of automatic reconstruction.
Turning it into a \emph{learning protocol} requires a workflow that (a) designs probe controls to generate informative control tensors $\mathcal{F}$, and (b) reconstructs the desired channel-indexed correlator sectors $\mathcal{M}^{\vec{\mu}}_{\vec{p}}(\vec{t})$, with matter entries $M_{m(p_j)}$, from measured probe data as an inverse problem.
Figure~\ref{fig:fig_1_demo} gives the circuit-level intuition: a driven probe register interacts with selected matter operators in successive time windows, and the final probe measurements act as the readout channel for the encoded correlator sectors.
The concrete protocol design and reconstruction details are given in {Methods} (Algorithm~1).


\section[4]{Entangled Quantum Probes: Learning Entanglement Entropy}
\label{sec:entropy}
A particularly striking manifestation of quantum probe-based learning arises when the goal is not just to characterise the overall fluctuations, but to reconstruct genuine quantum correlations of a many-body system. 
Among such quantities, the entanglement entropy of a subsystem provides a fundamental diagnostic of quantum matter, capturing correlations that have no classical analogue. 
Yet from an operational perspective, learning entanglement entropy is highly nontrivial: it is not a response function and cannot in general be inferred from susceptibilities alone. 
Entangled quantum probes provide a natural route to access the necessary information.

\subsection[ssec:free_matter]{Warm-up: Gaussian many-body systems.}

To build intuition, we first consider Gaussian bosonic systems whose reduced states are Gaussian.
This provides the cleanest starting point because Gaussianity compresses the entanglement-entropy reconstruction problem to first and second moments: the covariance matrix is a sufficient statistic, without the need for full state tomography or higher-order correlators.
{This is exactly the quasi-free, second-order tier highlighted in Fig.~\ref{fig:fig_2_probe_circuit}.}

Let $R_A=(q_1,p_1,\ldots,q_{|A|},p_{|A|})^T$ denote the vector of canonical quadratures restricted to a region $A$ as a subsystem of $M$.
The equal-time symmetrized correlators define the covariance matrix
$(V_A)_{ij}=\tfrac12\langle\{\Delta R_{A,i},\Delta R_{A,j}\}\rangle$,
with $\Delta R_{A,i}=R_{A,i}-\langle R_{A,i}\rangle$.
For Gaussian states this matrix completely determines the reduced density matrix $\rho_{A,G}$ up to displacements, which do not affect the entropy.
Its symplectic eigenvalues $\{\nu_\alpha\}$ are obtained from the spectrum of $i\Omega_A V_A$, where
$\Omega_A=\bigoplus_{\alpha=1}^{|A|}\begin{pmatrix}0&1\\-1&0\end{pmatrix}$.
These eigenvalues may be viewed as the effective occupation data of the independent normal modes obtained after bringing the Gaussian state on $A$ to its canonical form.
The von Neumann entropy of the subsystem takes the universal form
\begin{equation}
S(\rho_{A,G})=\sum_{\alpha=1}^{|A|}
\Big[(\bar n_\alpha+1)\ln(\bar n_\alpha+1)-\bar n_\alpha\ln\bar n_\alpha\Big],
\end{equation}
and $\bar n_\alpha={\nu_\alpha}-\frac12 $.
In practical terms, the Gaussian entropy-learning problem is therefore simple to state:
learn the covariance matrix $V_A$, extract the symplectic eigenvalues $\{\nu_\alpha\}$, and evaluate the formula above.
No higher-order correlators are needed in this warm-up setting.

The key point for our theme is operational.
The covariance matrix $V_A$ is built entirely from equal-time anti-commutators, namely from symmetrised fluctuations of local quadratures on $A$.
By contrast, the equal-time commutators are fixed by the canonical algebra and therefore do not contain state-specific information.
Thus, even in this simplest entropy problem, the relevant data lie in the fluctuation sector rather than in the response sector.


Quantum probes are naturally suited to this task. In this sense, the probe does not learn the entropy in one opaque step; it first learns the fluctuation data on $A$, from which the entropy follows as a known functional.
This Gaussian warm-up is therefore the simplest concrete instance of the main theme of the paper: access to the anti-commutator sector turns an entropy-related quantity that is invisible to system-only response into an operationally learnable one.

\subsection[ssec:interacting_matter]{Interacting many-body systems.}

For interacting matter the situation becomes qualitatively richer.
The reduced state $\rho_A$ is no longer Gaussian, so the entropy cannot be reconstructed from two-point correlators alone.
Instead, interactions make the entropy sensitive to higher-order many-body correlation data.
The problem is therefore no longer merely to learn a covariance matrix, but to learn the higher-order local correlators that correct the Gaussian reference state.

A convenient way to organize this structure is to expand the entropy around the Gaussian reference state $\rho_{A,G}$ generated by the quadratic part of the Hamiltonian,
\begin{equation}
S(\rho_A)=S(\rho_{A,G})+\sum_{n\ge 1}\delta S^{(n)} ,
\label{eq:S_series}
\end{equation}
where the correction $\delta S^{(n)}$ is determined by higher-order local correlation data in region $A$.
This equation makes the operational task transparent.
The Gaussian part $S(\rho_{A,G})$ is already fixed by the two-point fluctuation data discussed above, whereas the genuinely interacting contribution depends on higher-order local correlator data, often organized as connected cumulants on the entropy side.
The connected organization is not an additional probe observable: the probe protocol reconstructs the underlying operator-ordering sectors of multi-point correlators.

If we denote by $R_i(t)$ the local matter operators relevant to the subsystem $A$, then the required multi-point data take the form $\langle R_{i_1}(t_1)\cdots R_{i_m}(t_m)\rangle$.
The key structural point is that such correlators are not single objects from the perspective of probe readout.
Rather, they decompose into the full family of braketor sectors introduced in Sec.~\ref{sec:probe_landscape}:
\begin{equation}
\expect{R_{i_1}(t_1)\cdots R_{i_m}(t_m)}
=
\sum_{\vec{\mu}}\mathcal{M}^{\vec{\mu}}_{\vec{i}}(\vec{t}),
\end{equation}
with $t_1\ge\cdots\ge t_m$.
Here $\vec{i}=(i_1,\ldots,i_m)$ is only a target-operator shorthand for the desired matter correlator; the corresponding probe-channel implementation is made explicit in the Methods.
This identity is precisely why the probe framework is useful in the interacting setting: the higher-order correlators needed for entropy reconstruction are assembled from the same operator-ordering sectors that probe readout can access.

The distinction from system-only response becomes sharp from here.
Conventional response theory accesses only the fully retarded commutator sector.
That single sector is generally insufficient to reconstruct the full multi-point correlators, and hence the connected information that determines interacting entanglement entropy.
Quantum probes, by contrast, access the enlarged family of sectors $\mathcal{M}^{\vec{\mu}}$, providing exactly the missing operator-ordering data.

The interacting entanglement entropy becomes operationally learnable not because the probe directly ``measures'' entropy, but because it gives access to the higher-order correlator sectors from which the entropy-side connected quantities are built.
The important point is that probe-based learning remains aligned with the central theme of this work: once the probe is treated as a measured open quantum subsystem, the information channel is enlarged beyond the response sector, and that enlargement is precisely what interacting entropy requires.

\subsection{How many entangled probes and how many settings are needed?}

We can now return to the second question raised in the Introduction and sharpened in Sec.~\ref{sec:probe_framework}: if interacting entropy requires higher-order correlator data, how large a probe register is actually required?
At this stage the answer is conceptually clear.
The probe-learning framework resolves the correlator hierarchy order by order, so the relevant resource is set by the highest operator-insertion order one aims to access, not by the total size of the many-body system.

The reason is that each additional order in the entropy expansion corresponds to one additional operator insertion in the correlator data to be learned.
Operationally, this means that accessing a $k$-point correlator requires a probe architecture capable of coherently encoding $k$-point operator structure in the final readout.
Consequently, the probe overhead tracks correlator complexity rather than system size.

\begin{corollary}[Probe count tracks correlator order]
To access matter correlator data up to order $K$, and hence connected quantities of that order when an entropy expansion uses them, it suffices to employ $\mathcal{O}(K)$ coherently controlled entangled probes, independent of the system size.
\label{theo:entanglement_entropy}
\end{corollary}


This is the precise sense in which entangled probes are parsimonious learners of interacting many-body structure.
The probe register need not mirror the full many-body system.
It needs only match the complexity of the target information to be learned.
That conclusion directly realises the programme announced previously: quantum probes enlarge learnability without requiring a probe that scales with the system itself.
A formal worst-case proof, including the distinction between direct addressing and possible probe reuse, is given in the Supplementary Material.

A separate resource question concerns not the size of the probe register in a single run, but the number of distinct experimental settings needed to reconstruct the entropy of a chosen subsystem.
Here locality matters.
Because the entropy expansion involves local correlator data supported within region $A$, the relevant setting count is controlled by the subsystem size and the perturbative order, rather than by the size of the full matter system.

\begin{corollary}[Experimental-setting count for interacting entropy learning]
To learn the von Neumann entropy of an interacting subsystem $A$ up to $k$th-order perturbation, it suffices to use at most $\mathcal{O}(|A|^{k})$ distinct experimental settings.
Each setting involves $\sim k$ entangled probes, where the probes couple to operators supported within the subsystem $A$.
\end{corollary}
The corresponding counting argument and locality assumptions are stated explicitly in the Supplementary Material.


We would like to point out the comparison with tomography-based entropy estimation. The first distinction is apparently operational: Probe-based learning asks what can be learned when information is mediated through a controllable quantum interface and a final probe readout, whereas tomography-based learning requires measurements on the system. From the computational complexity point of view, the sample complexity by tomography or classical shadows for purity (second-order Renyi entropy) scales $N_{\rm sh}=\mathcal O(4^{|A|}/\epsilon^2)$ up to logarithmic and error-propagation factors, and thus scales exponentially in $|A|$. In stark contrast, the probe route uses $N_{\rm probe}^{\rm set}\sim\mathcal O(|A|^k)$ settings, with $\mathcal O(k)$ probes per setting, because it learns the $k$th-order correlator manifold relevant to the chosen expansion.
In probe-based learning, the idea is that any physical property that can be expressed through a controlled family of such correlators---fluctuations, covariance data, dynamical response beyond the retarded sector, non-equilibrium correlation structure, order-parameter correlations, and entropy-related quantities---can be targeted without reconstructing the full many-body state. The comparison is not that one paradigm subsumes the other; it is that coherent probes provide a complementary and often more natural route whenever the desired physics is organised by low-order or local correlator structure rather than by full-state reconstruction.

Taken together, these results complement the picture for interacting systems.
The advantage of entangled probes arises from their ability to be tailored to the order and support of the correlations that govern the quantity of interest. Therefore, for interacting entanglement entropy, the relevant resource is determined by the information complexity—namely, the correlator order and subsystem support—rather than by the size of the underlying many-body system.

\section[5]{Discussion}

By treating the probe as a measured quantum subsystem rather than only as a source of perturbation, we have exposed an operator-level boundary of conventional response theory.
System-only response organizes information into fully retarded nested commutators.
Reduced probe dynamics, by contrast, carries a broader braketor family generated by coherent probe--matter evolution and final probe readout.
The distinction is therefore not one of sensitivity alone, but of algebra: quantum probes open correlator sectors that are absent from response data except under additional constraints such as equilibrium fluctuation--dissipation relations. 

This algebraic enlargement has concrete consequences.
A single coherently controlled probe already separates the phase channel, tied to response, from the decoherence channel, tied to fluctuations.
For quasi-free matter, the symmetrised sector supplies the covariance data that determine entropy-related observables.
For interacting matter, the same logic extends to higher-order local correlator data, so the relevant auxiliary resource tracks the order and support of the target correlations rather than the Hilbert-space dimension of the full many-body system.

\sun{We have illustrated the connection and comparison with established matter-probe techniques. We would like to further compare notions in quantum technologies.} The advantage we study here is conceptually distinct from entanglement-enhanced metrology.
The latter asks how entanglement improves sensing {precision} (e.g., Heisenberg-limited scaling) for a field parameter that is already measurable~\cite{liu2020quantum, greve2022entanglement,higgins2007entanglement,leibfried2004toward,wang2019heisenberg,deng2024quantum}. 
We instead focus on {capability}: whether quantum (and entangled) probes enlarge the set of qualitative features of many-body systems that can be operationally learned at all. \sun{It is also easy to see the difference from Hamiltonian learning (by compressed sensing or probe tomography~\cite{huang2023learning,hu2025ansatz,chen2025quantum}) as it learns the properties rather than the Hamiltonian parameter (which are usually real numbers).} 
Quantum noise spectroscopy exemplifies how controlled quantum probes can infer spectra of otherwise inaccessible environmental degrees of freedom~\cite{szankowski2016spectroscopy,yan2013rotating,norris2016qubit,yoneda2023noise}. Here we move from environmental characterisation to a theory of many-body feature detection: by identifying the operator algebra exposed by probe readout beyond system-only response, we determine which physical properties can be learned and the quantum-probe resources required to learn them.

Our results place probe-based learning alongside tomography-based quantum learning, but with a different strength. Tomography and classical shadows start from direct access to the target degrees of freedom and are designed to reconstruct, or compressively predict, properties of the state under randomised measurements.
In probe-based, one controls an auxiliary quantum system, couples probes to the chosen subsystem of matter, and post-processes the measurement outcome to infer either observable expectations or entropy.
Spatially structured many-body properties, such as entanglement entropy, cannot in general be reduced to a collection of independent local signals. Their distributed information must instead be coherently mapped onto collective degrees of freedom of the probe register. Tunable probe entanglement provides these collective readout channels, allowing the underlying many-body structure to be resolved and reconstructed.


The practical message is that probe advantage need not require an auxiliary device comparable in size to the target.
What is required is a controllable probe register whose interaction history with the matter can be programmed and whose final reduced state can be measured.
For platforms that already possess coherent ancillae, this points less to new hardware than to new protocols for using the same hardware as an information-bearing quantum interface.
More broadly, quantum probes should be viewed not merely as quantum sensors, but as coherent gateways to sectors of many-body physics that conventional response leaves unseen.





\widetext
\section*{Methods}

\subsection{Comparison between system-based and probe-based learning}

{We first discuss the operational distinction between quantum and classical probes. A classical probe is not assumed to be microscopic or macroscopic; rather, it is a probe whose Hilbert-space degree of freedom is not retained as a readout register. If embedded in a formal probe--matter space, this access model allows only a source term of the form $(\lambda(t)I_P)\otimes M(t)$, up to probe-local evolution. Indeed, if the probe carries no readable memory of the matter, the joint evolution must factorise as $U_{PM}(T)=U_P(T)\otimes U_M^\lambda(T)$, so its generator is probe-local plus matter-local, with $U_M^\lambda(T)$ generated by $H_M+\lambda(t)M(t)$. 

{Take typical inelastic neutron scattering as an example. The probe density matrix of incoherent neutron $\hat\rho_{\mathrm{p}}$ is
\emph{diagonal} in the momentum basis:
$ 
  \hat\rho_{\mathrm{p}}
  = \int d\mathbf{k}_i\, P(\mathbf{k}_i)\,
    \ket{\mathbf{k}_i}\bra{\mathbf{k}_i}.
$
There are no fixed phase relations between different momenta. The transition probability from an initial state
$\ket{i} $ to a final state
$\ket{f}  $ is  given by Fermi's golden rule, from which one can derive the differential cross section and the dynamic structure factor.
All explicit probe operators disappear,
and the remaining factor defines the {dynamic structure factor}
$
  S(\mathbf{q},\omega)
  =
  \frac{1}{2\pi}
  \int_{-\infty}^{\infty} dt\;
  e^{i\omega t}
  \big\langle
    \hat O_{\mathrm{s}}^\dagger(\mathbf{q},0)
    \hat O_{\mathrm{s}}(\mathbf{q},t)
  \big\rangle.
 $
Hence, for incoherent probes such as neutrons,
the probe degrees of freedom can be fully traced out,
and the scattering cross section depends only on sample
correlation functions.
}

By contrast, an interaction $h_P(t)\otimes M(t)$ with $h_P(t)\not\propto I_P$ produces probe-conditioned matter evolution; for a probe superposition, this generically creates probe--matter entanglement and leaves matter information in the final reduced probe state. This is the operational origin of the algebraic distinction summarised in Table~\ref{tab:probe_vs_system_correlators}.}

\begin{table}[H]
\centering
\begingroup
\renewcommand{\arraystretch}{1.25}
\setlength{\tabcolsep}{4pt}
\resizebox{\textwidth}{!}{%
\begin{tabular}{|c|c|c|c|}
\hline
 & Learnable correlators & $\#$ of combinations & Linear term \\
\hline
\begin{tabular}{c}
\textbf{System-based} \\
(closed system)
\end{tabular}
&
$\mathcal{M}^{\vec{1}}(t_{[k]}) \sim \expect{[\cdots[ M_{m_1}(t_1), M_{m_2}(t_2) ], \ldots , M_{m_k}(t_k)]}$
&
$1$
&
$\expect{\big[ M_{m_1}(t_1),\, M_{m_2}(t_2) \big] }$
\\
\hline
\begin{tabular}{c}
\textbf{Probe-based} \\
(open system)
\end{tabular}
&
\begin{tabular}{c}
$\mathcal{M}^{\vec{\mu}}(t_{[k]}) \sim
\expect{\llbracket\cdots \llbracket
M_{m_1}(t_1),M_{m_2}(t_2)\rrbracket_{\mu_1},
\cdots, M_{m_k}(t_k)\rrbracket_{\mu_{k-1}} } $\\[6pt]
$\vec{\mu}\in\{0,1\}^{k-1}$
\end{tabular}
&
$2^{k-1}$
&
\begin{tabular}{c}
$\expect{\big[ M_{m_1}(t_1),\, M_{m_2}(t_2) \big]}/2$
\\[4pt]
$ \expect{\{ M_{m_1}(t_1),\, M_{m_2}(t_2) \} }/2$ \\
\end{tabular}
\\
\hline
\end{tabular}
}
\endgroup
\caption{Closed-system response vs.\ open-system probe learning at the level of operator orderings.
System-only response is restricted to the fully retarded (nested-commutator) sector, whereas probe readout contains all braketor sectors,
including symmetrized (anti-commutator) structures that encode fluctuations. Here we drop the probe labelling for convenient comparison with system-only case (probe-channel indices and operator $\widehat{P}$ are formally retained elsewhere).
The $k=2$ ``linear term'' illustrates the first appearance of both components in probe-based learning.}
\label{tab:probe_vs_system_correlators}
\end{table}

\begin{figure*}[htbp]
\centering
\includegraphics[width=0.8\textwidth]{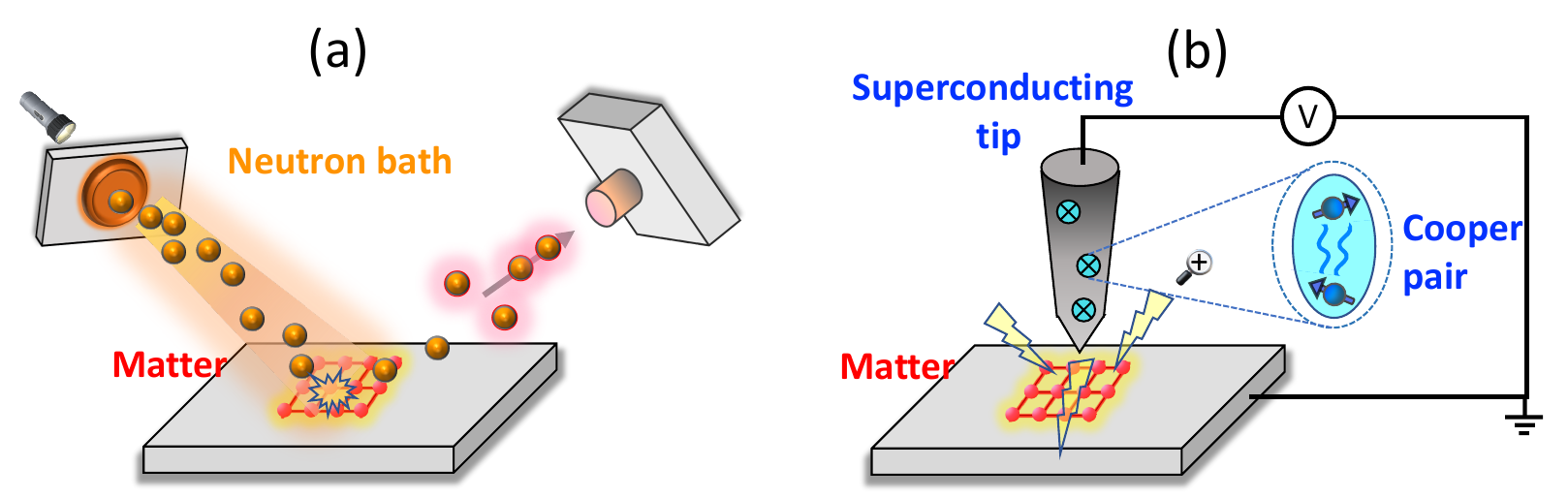}
\caption{\textbf{Experimental setup of probe-based learning as special cases within our framework.}
(a) In a conventional beam- or bath-style probe, illustrated here by neutron scattering from matter, the external particles are usually treated as an imposed perturbation and scattering readout for system-only response: the measured signal is interpreted through matter susceptibilities and therefore through retarded commutator structures. This picture could be extended to non-linear spectroscopy.
(b) A superconducting tip provides a different kind of interface.
The tunnelling object can be a phase-coherent Cooper pair, so the probe itself is a quantum degree of freedom that can become entangled with the target before its reduced state is read out through electrical observables.
This turns the probe--matter contact into an open-system learning channel rather than only a perturb-and-measure response experiment.
 }
\label{fig:fig_3_cooper_pair}
\end{figure*}


Figure~\ref{fig:fig_3_cooper_pair} translates the algebraic distinction in Table~\ref{tab:probe_vs_system_correlators} into a language closer to experiment.
In a scattering or spectroscopy mindset, the external probe is often integrated out conceptually: it prepares a perturbation, the matter evolves, and the recorded signal is interpreted as a response function of the target alone.
That viewpoint is powerful, but it naturally privileges the retarded sector.
Probe-based learning asks for a more quantum use of the same experimental idea: keep the probe as a controllable subsystem, calibrate its coupling history, and use its final reduced readout as the data channel.
The payoff is thus not only a stronger signal, but a larger operator alphabet for what the experiment can learn.

The superconducting-tip sketch in Fig.~\ref{fig:fig_3_cooper_pair}(b) is another motivation.
A Cooper pair carries a coherent two-particle quantum structure, and a biased or gated tip can in principle supply tunable coupling windows, phases, and readout channels.
Those knobs are precisely what appear below as the probe-side control tensor $\mathcal F$: they decide which linear combinations of matter operator orderings are encoded in the measured probe observable.
Thus, for experimental readers, the formalism below can be read as a design language for converting familiar probe hardware into a correlator-selective quantum interface, without requiring full tomography of the many-body target.  Further experimental design may be possible by considering controlled operations between the probe and matter.

\subsection{Open quantum dynamics of the probe-matter system}
\label{sec:methods_open_dynamics}

This subsection provides a minimal, self-contained bridge between the main-text structural Eq.~\eqref{eq:dyson_commutator} and the detailed algebra collected in the Supplementary Material.
The probe--matter coupling $H_{PM}(t)$ is defined in the main text; here we only record the interaction-picture objects and the resulting Dyson expansion of the probe readout.

Let $U_P(t)$ denote the propagator generated by $H_P(t)$; working in the interaction picture with respect to $H_P(t)+H_M$, the probe--matter interaction becomes
\begin{equation}
H_I(t)
= \sum_p \tilde h_p(t)\otimes M_{m(p)}(t),
\label{eq:methods_HI}
\end{equation}
where $\tilde{h}_p(t) \equiv U_P^\dagger(t) h_p(t) U_P(t)$ is the Heisenberg-evolved probe-side toggled operator set by the control.
All probe information about $M$ enters through $H_I(t)$ and the initial probe state $\rho_P$.

For a chosen probe observable $O_P$ measured at the final time $T$, the probe signal can be written as a Dyson series in $H_I$.
Using the shorthand $\int d_{>}t_{[k]}:=\int_{0}^{T}dt_1\cdots\int_{0}^{t_{k-1}}dt_k$ for strictly time-ordered integration ($t_1>\cdots>t_k$), one obtains
\begin{equation}
\label{eq:methods_dyson_raw}
\begin{aligned}
\langle O_P(T)\rangle_{\rho_P}
&=
\sum_{k\ge 0}\,(-i)^k
\sum_{\vec p}\sum_{l=0}^{k}\sum_{\pi\in\Pi_{l,k}}
\int d_{>}t_{[k]}\;
(-1)^l\,
\Tr_P\!\Bigg[
\prod_{j=1}^{l}\tilde h_{p_j}(t_{\pi(j)})\, O_P\,
\prod_{j'=l+1}^{k}\tilde h_{p_{j'}}(t_{\pi(j')})\,\rho_P
\Bigg] \\
&\quad \times
\Big\langle
M_{m(p_1)}(t_{\pi(1)})\cdots M_{m(p_k)}(t_{\pi(k)})
\Big\rangle.
\end{aligned}
\end{equation}
Here $\vec p=(p_1,\ldots,p_k)$ labels the probe-channel slots appearing in $H_I(t)$; the matter operator in slot $p_j$ is $M_{m(p_j)}$.
The index $l$ counts how many probe operators act to the left of $O_P$ inside the trace, and the set $\Pi_{l,k}$ collects the corresponding forward/backward time allocations.
For each branch assignment $\pi$, define the associated ordered matter string
\begin{equation}
\label{eq:methods_C_def}
\mathcal C^{(k)}_{\vec p,\pi}(t_{[k]})
:=
\Big\langle
M_{m(p_1)}(t_{\pi(1)})\cdots M_{m(p_k)}(t_{\pi(k)})
\Big\rangle,
\end{equation}
where the Supplementary Material defines the allowed branch assignments $\pi\in\Pi_{l,k}$ and the reduced branch label $\bm{\pi}_M$ in a binary representation associated with each $\pi$.

The string correlator $\mathcal C$ can be expanded in the basis of $(k\!-\!1)$-level \emph{nested bracketors},
labelled by $\vec\mu=(\mu_1,\ldots,\mu_{k-1})\in\{0,1\}^{k-1}$, where $\mu_j=1$ denotes a commutator insertion and $\mu_j=0$ an anti-commutator insertion.
For each branch assignment, one has
\begin{equation}
\label{eq:methods_nested_rep}
\mathcal C^{(k)}_{\vec p,\pi}(t_{[k]})
=
\sum_{\vec\mu\in\{0,1\}^{k-1}}
(-1)^{\vec\mu\cdot \bm{\pi}_M}\,
\mathcal{M}^{\vec\mu}_{\vec p}\!\big(t_1,\ldots,t_k\big),
\end{equation}
where $\mathcal{M}^{\vec\mu}_{\vec p}$ is the corresponding nested commutator/anti-commutator correlator defined in Eq.~\eqref{eq:dyson_commutator}; its subscript is the probe-channel tuple, while its matter entries are $M_{m(p_j)}$.
Substituting Eq.~\eqref{eq:methods_nested_rep} into Eq.~\eqref{eq:methods_dyson_raw} and collecting all probe-side factors defines the \emph{control tensor}: 
\begin{equation}
\label{eq:methods_F_def}
\mathcal{F}^{(k;\vec\mu)}_{O_P;\vec p;t_{[k]}}[H_P(\vec{t})]
:=
\sum_{l=0}^{k}\sum_{\pi\in\Pi_{l,k}}(-1)^{l+\vec\mu\cdot \bm{\pi}_M}\,
\Tr_P\!\Bigg[
\prod_{j=1}^{l}\tilde h_{p_j}(t_{\pi(j)})\, O_P\,
\prod_{j'=l+1}^{k}\tilde h_{p_{j'}}(t_{\pi(j')})\,\rho_P
\Bigg].
\end{equation}
In this representation, \emph{all} system dependence is isolated in the nested-bracket correlators $ \mathcal{M}^{\vec\mu}_{\vec p}$, while \emph{all} experimental design freedom (probe state, readout, and control history) is isolated in $\mathcal{F}$.
The Supplementary Material provides the detailed derivation, including the branch-string basis $\mathcal C$, the braketor basis $\mathcal{M}^{\vec\mu}_{\vec p}$, their Walsh--Hadamard relation, and the constructive windowed design illustrated in Supplementary Fig.~S1.

\subsection{Quantum-probe workflow: turning readout into correlators}

Theorem~\ref{theo:probe_adv} becomes a \emph{learning protocol} once the time-ordered convolution in Eq.~\eqref{eq:dyson_commutator} is made experimentally identifiable.
The protocol does not require a brute-force inversion of the full high-dimensional kernel.
The clearest implementation is a windowed coupling: each selected probe--matter channel is active only in a short time window around a target insertion time, while unused channels are switched off.
In the ideal narrow-window limit, this localizes the convolution and turns the order-$k$ readout into a calibrated sample of an ordered matter string.
Always-on or broad-window couplings obey the same inverse-problem logic, but with less localized kernels that must be deconvolved, as detailed in Supplementary Sec.~S3 and Supplementary Fig.~S1.

Fix a target order $k$, probe slots $\vec p=(p_1,\ldots,p_k)$, matter labels $\vec m=(m_1,\ldots,m_k)$, and ordered times $\bar\tau=(\tau_1,\ldots,\tau_k)$ with $\tau_1\ge\cdots\ge\tau_k$.
A window setting $\eta$ specifies two pieces of hardware control: the active interaction windows and a run-specific channel-to-matter map $p\mapsto m_\eta(p)$.
For a selected slot $p_j$, one may write the controlled probe-side operator as $$\tilde h_{p_j}^{[\eta]}(t)\simeq f_\sigma(t-\tau_{\eta(j)})\tilde h_{p_j}(t)$$ where the window envelope approaches $\delta(t-\tau_{\eta(j)})$ as $\sigma\to0$; unused channels have $\tilde h_p^{[\eta]}(t)=0$.
The map $m_\eta$ is fixed within that run; only the coupling envelope is time-windowed.
Schematically,
\begin{equation}
\label{eq:methods_windowed_selection}
H_I^{[\eta]}(t)
\simeq
\sum_{p\in P}\tilde h_p^{[\eta]}(t)\otimes M_{m_\eta(p)}(t),
\end{equation}
where $\tilde h_p^{[\eta]}(t)$ includes the window envelope and unused channels are switched off.
The fixed-channel setting of Eq.~\eqref{eq:general_coupling} is the special case in which all runs use the same map, $m_\eta(p)=m(p)$.
Substituting Eq.~\eqref{eq:methods_windowed_selection} into the Dyson expansion gives, in the narrow-window limit,
\begin{equation}
\label{eq:methods_windowed_readout}
\langle O_P(T)\rangle_{\alpha,[\eta]}^{(k)}
\simeq
(-i)^k\,
G^{(k)}_{\alpha;\vec p,\bm{\eta}_M}(\bar\tau)\,
\mathcal C_{\vec p,\bm{\eta}_M}^{(k)}(\bar\tau),
\end{equation}
up to finite-width corrections.
Here $\alpha$ labels the probe setting, namely the preparation, probe control, and readout basis.
The known scalar $G^{(k)}_{\alpha;\vec p,\bm{\eta}_M}$ is the calibration factor for that windowed run: it is fixed entirely by the probe setting and applied windows, while the unknown matter information sits in $\mathcal C$.
Its trace-level expression is given in the Supplementary Material.
Repeating the experiment with one window pattern for each reduced matter-branch label $\bm\nu$ reconstructs the strings $\mathcal C_{\vec p,\bm\nu}^{(k)}(\bar\tau)$.
The desired braketor sectors are then obtained by the inverse binary transform
\begin{equation}
\label{eq:methods_C_to_M}
\mathcal M_{\vec p}^{\vec\mu}(\bar\tau)
=
2^{-(k-1)}
\sum_{\bm\nu\in\{0,1\}^{k-1}}
(-1)^{\vec\mu\cdot\bm\nu}\,
\mathcal C_{\vec p,\bm\nu}^{(k)}(\bar\tau),
\end{equation}
where each reduced label $\bm\nu$ is realized experimentally by choosing a window pattern $\eta$ with $\bm{\eta}_M=\bm\nu$, as defined in the Supplementary Material.
For $k=2$, Supplementary Fig.~S1(b,c) shows the two patterns: $\eta_0$ isolates $\mathcal C_{(p_1,p_2),(0)}^{(2)}=\langle M_{m_1}(\tau_1)M_{m_2}(\tau_2)\rangle$, whereas $\eta_1$ isolates $\mathcal C_{(p_1,p_2),(1)}^{(2)}=\langle M_{m_1}(\tau_2)M_{m_2}(\tau_1)\rangle$ in the canonical probe-slot bookkeeping.

Operationally, the window setting chooses the sampled matter string, while the probe setting (initial state, control unitary, and measurement basis) fixes the calibrated coefficient with which it appears in readout.
After windowing, or deconvolution for broader kernels, the data are linear equations for the unknown string amplitudes; the probe unitaries are not unique, and any full-rank, well-conditioned calibrated family can reveal the desired correlator or braketor sector.

\begin{algorithm}[H]
\label{alg:Alg1}
\caption{Windowed quantum-probe learning of braketor correlators}
\begin{algorithmic}[1]
\State \textbf{Input:} Target order-$k$ object $(\vec p,\vec m,\bar\tau)$ and desired sector labels $\vec\mu$.
\State \textbf{Output:} Estimates $\widehat{\mathcal{M}}^{\vec{\mu}}_{\vec{p}}(\bar\tau)$.

\State \textbf{(1) Route the matter strings.}
For each reduced label $\bm\nu$ in Eq.~\eqref{eq:methods_C_to_M}, choose a window setting $\eta$ whose active windows place the probe slots at $\bar\tau$ and whose fixed run map couples them to $\vec m$.
\State \textbf{(2) Choose readable probe settings.}
Choose initial states, unitaries, and measurement bases whose calibrated readout coefficients give a full-rank linear system.
\State \textbf{(3) Use the simplest full-rank family.}
Prefer native or Clifford unitaries, product or weakly entangled probe inputs, and local measurement bases; enlarge the family only to improve conditioning.
\State \textbf{(4) Estimate ordered strings.}
Run the settings, isolate the order-$k$ contribution, and solve for $\widehat{\mathcal C}_{\vec p,\bm\nu}^{(k)}(\bar\tau)$.
\State \textbf{(5) Convert strings to braketors.}
Apply Eq.~\eqref{eq:methods_C_to_M}; use redundant windows or probe settings for error estimates and consistency checks.
\end{algorithmic}
\end{algorithm}

Algorithm~1 serves as an experimental design checklist.
The reduced label $\bm\nu$ names one of the ordered strings entering Eq.~\eqref{eq:methods_C_to_M}, while $\eta$ names the laboratory setting that realises it through active windows and the fixed run map $p\mapsto m_\eta(p)$.
Steps (2)--(3) use the remaining probe freedom: different preparations, unitaries, and local readout bases give different calibrated coefficients, so the only requirement is a full-rank, well-conditioned linear system for the selected strings.
Step (4) extracts the desired perturbative order, for example, by coupling-strength scaling or phase cycling, and Step (5) performs the fixed algebraic recombination from ordered strings to braketor sectors.
Thus the protocol is correlator-selective learning, not many-body state tomography; finite-width and always-on variants give the same linear-identification problem after deconvolution, with the explicit windowed constructions and error terms detailed in Supplementary Sec.~S3.

\subsection{Entropy construction from many-body correlators}
\label{sec:methods_entropy}

This subsection outlines the operator structure underlying the entropy reconstruction discussed in Sec.~\ref{sec:entropy}. 
Let $\rho_A$ denote the reduced density matrix of a subsystem $A$ of the matter system. 
The von Neumann entropy is $S(\rho_A)=-\Tr_A(\rho_A\log\rho_A)$. 
In the situations considered in the main text, $\rho_A$ can be written as a perturbative expansion around a Gaussian reference state,
$\rho_A=\rho_A^{(0)}+\delta\rho_A$, where $\rho_A^{(0)}$ is determined by the two-point correlations of the matter operators and $\delta\rho_A$ encodes interaction-induced corrections. 
Expanding the entropy functional around $\rho_A^{(0)}$ generates a series 
$S(\rho_A)=S^{(0)}+\sum_{n\ge1}\delta S^{(n)}$, where each correction $\delta S^{(n)}$ can be expressed in terms of many-body correlators of the matter operators.

Let $R(x,t)$ denote the Heisenberg-picture matter operators appearing in the $P-M$ coupling, namely the operators $M_{m(p)}(t)$ selected by the channel map introduced in Sec.~\ref{sec:probe_landscape}.
At equal time they reduce to the spatial operators $R(x)$. 
The entropy corrections depend on equal-time correlators of these operators, schematically

\begin{equation}
\delta S^{(n)}
\sim
\int dx_1\cdots dx_m
\;
\langle
R(x_1,t)\cdots R(x_m,t)
\rangle ,
\end{equation}
where $x$ labels spatial coordinates and the kernels multiplying these correlators depend only on the subsystem geometry $A$.

The operator products appearing in these correlators can be decomposed into the same nested braketor structures introduced in Sec.~\ref{sec:probe_landscape}. 
In particular,

\begin{equation}
\langle
R(x_1,t_1)\cdots R(x_m,t_m)
\rangle 
=
\sum_{\vec\mu}
\mathcal{M}^{\vec\mu}_{\vec{x}}(t_1,\ldots,t_m) ,
\end{equation}
where $\mathcal{M}^{\vec\mu}$ denotes the nested commutator/anti-commutator correlators defined in the main text and $\vec{\mu}\in\{0,1\}^{m-1}$ labels the braketor sector.
The subscript $\vec{x}=(x_1,\ldots,x_m)$ is a target-space shorthand used only in this entropy construction, because the entropy kernels below are functions of matter coordinates.
In the operational probe notation of Eq.~\eqref{eq:dyson_commutator}, the same object is written as $\mathcal{M}^{\vec\mu}_{\vec p}$ after selecting probe channels $\vec p=(p_1,\ldots,p_m)$ such that $M_{m(p_j)}(t_j)=R(x_j,t_j)$.
Thus $\vec{x}$ labels the matter correlator being reconstructed, whereas $\vec p$ records the probe-channel history that realizes it.
These structures include both fluctuation correlators ($\mathcal{M}^{+}$) and response correlators ($\mathcal{M}^{-}$).

The spatial correlators determining the entropy correspond to the equal-time limit of these multi-time structures ($t_1=\cdots=t_m=t$). 
Consequently the entropy can be written schematically as 
\begin{equation}
  S(\rho_A)=S^{(0)}+\sum_m\int dx_1\cdots dx_m\,\mathcal{K}_m(x_1,\ldots,x_m)\mathcal{M}^{\vec\mu}_{\vec{x}}(t,\ldots,t),
\end{equation}
where the kernels $\mathcal{K}_m$ depend only on the subsystem $A$.

Because the probe dynamics derived in Sec.~\ref{sec:probe_landscape} provides direct experimental access to the expectation values $\langle\mathcal{M}^{\vec\mu}\rangle$, the entropy of interacting matter systems can therefore be reconstructed from probe measurements. 
Further details of the entropy expansion and kernel construction are provided in the Supplementary Material.

\begin{acknowledgments}
We would like to thank Tianfeng Feng for useful discussions. We acknowledge funding from the EPSRC through EP/S021582/1 and from Innovate UK (Project No.10075020).
\end{acknowledgments}

\bibliography{Ref}
\bibliographystyle{unsrt}


\clearpage

\setcounter{section}{0}
\setcounter{subsection}{0}
\setcounter{equation}{0}
\setcounter{figure}{0}
\setcounter{table}{0}

\renewcommand{\thesection}{S\arabic{section}}
\renewcommand{\thesubsection}{\Alph{subsection}}
\renewcommand{\theequation}{S\arabic{equation}}
\renewcommand{\thefigure}{S\arabic{figure}}
\renewcommand{\thetable}{S\arabic{table}}


\begin{center}
  {\bfseries Supplementary Materials\par}
\end{center}

\vspace{1.2em}


\titlecontents{section}
  [0pt]
  {\addvspace{0.55ex}}
  {\contentslabel[\thecontentslabel.]{2.4em}}
  {\hspace*{-2.4em}}
  {\hfill\contentspage}

\titlecontents{subsection}
  [2.6em]
  {}
  {\contentslabel[\thecontentslabel.]{1.7em}}
  {\hspace*{-1.7em}}
  {\hfill\contentspage}

\startcontents[supplement]

\begin{center}
\begin{minipage}{0.84\linewidth}
  \printcontents[supplement]{}{1}[2]{}
\end{minipage}
\end{center}

\vspace{1.2em}

The notation of the main manuscript is followed throughout the Supplementary Information.
In particular, the binary bracket
\begin{equation}
\llbracket X,Y\rrbracket_{\mu}:=XY+(-1)^{\mu}YX,
\qquad \mu\in\{0,1\},
\end{equation}
is normalized so that $\mu=1$ denotes a commutator and $\mu=0$ denotes an anti-commutator.

\section{General algebraic structure of quantum-probe learning}
\label{sec:supp_general_algebra}
This section supplies the mathematical backbone of the main-text claim that probe readout is an open-system object and therefore resolves a larger operator-ordering algebra than system-only response.
The role of this section is to isolate the pieces that are actually needed for the present manuscript:
the forward--backward Dyson expansion, the definition of the ordered branch assignments, the nested-braketor basis, and the explicit control tensor that separates probe design from matter information.

\subsection{Interaction-picture setup and ordered Dyson expansion}
\label{subsec:supp_probe_expansion}
Let the full Hamiltonian be
\begin{equation}
H(t)=H_P(t)+H_M+H_{PM}(t),
\qquad
H_{PM}(t)=\sum_{p\in P} h_p(t)\otimes M_{m(p)} ,
\end{equation}
where $P$ denotes the probe register, $M$ the target many-body system, and $p\mapsto m(p)$ the fixed assignment from probe channels to the matter operators selected for the experiment.
Working in the interaction picture with respect to $H_P(t)+H_M$, we write
\begin{equation}
H_I(t)=\sum_{p\in P} \tilde h_p(t)\otimes M_{m(p)}(t),
\qquad
\tilde h_p(t):=U_P^\dagger(t)\,h_p(t)\,U_P(t),
\label{eq:toggled_PM_H}
\end{equation}
with $U_P(t)$ the propagator generated by $H_P(t)$ alone.
Starting from the factorized state $\rho_P\otimes \rho_M$, the final probe signal is
\begin{equation}
\label{eq:supp_probe_signal_exact}
\langle O_P(T)\rangle_{\rho_P}
=
\Tr_{PM}\!\left[
(O_P\otimes \mathbb{I}_M)\,
U_I(T)\,(\rho_P\otimes \rho_M)\,U_I^\dagger(T)
\right],
\end{equation}
where
\begin{equation}
U_I(T)=\T\exp\!\left(-i\int_0^T dt\,H_I(t)\right).
\end{equation}

The essential structural point is that Eq.~\eqref{eq:supp_probe_signal_exact} contains both $U_I(T)$ and $U_I^\dagger(T)$.
After expanding them separately, one obtains forward- and backward-branch operator strings.
This is the source of the enlarged algebra.

To make this explicit, write
\begin{equation}
\label{eq:supp_UI_forward_backward}
\begin{aligned}
U_I(T)
&=
\sum_{r\ge 0}
(-i)^r
\int_0^T dt_1\cdots dt_r\;
\T\!\big[H_I(t_1)\cdots H_I(t_r)\big],\\
U_I^\dagger(T)
&=
\sum_{s\ge 0}
(+i)^s
\int_0^T dt'_1\cdots dt'_s\;
\bar{\T}\!\big[H_I(t'_1)\cdots H_I(t'_s)\big].
\end{aligned}
\end{equation}
Substituting Eq.~\eqref{eq:supp_UI_forward_backward} into Eq.~\eqref{eq:supp_probe_signal_exact} gives the branch-resolved expansion
\begin{equation}
\label{eq:supp_branch_resolved}
\begin{aligned}
\langle O_P(T)\rangle_{\rho_P}
=
\sum_{r,s\ge 0}
(-i)^r(i)^s
\int_0^T dt_1\cdots dt_r
\int_0^T dt'_1\cdots dt'_s\;
\Tr_{PM}\!\Big[
\bar{\T}\!\big(H_I(t'_{[s]})\big)\,
(O_P\otimes \mathbb I_M)\,
\T\!\big(H_I(t_{[r]})\big)\,
(\rho_P\otimes \rho_M)
\Big],
\end{aligned}
\end{equation}
where $t_{[r]}=(t_1,\ldots,t_r)$ and $t'_{[s]}=(t'_1,\ldots,t'_s)$.
Equation~\eqref{eq:supp_branch_resolved} is the most direct mathematical form of the forward/backward story.
Its drawback is only notational:
for later algebraic manipulations, it is more convenient to merge the two time lists into one globally ordered list.

For a fixed total order $k$, it is convenient to integrate over ordered times
\begin{equation}
\int d_>t_{[k]}
:=
\int_0^T dt_1\int_0^{t_1}dt_2\cdots\int_0^{t_{k-1}}dt_k ,
\qquad
t_1\ge t_2\ge \cdots\ge t_k .
\end{equation}
One then distributes the $k$ ordered times between the left and right of the measured probe observable $O_P$.
For each $l\in\{0,\ldots,k\}$, define $\Pi_{l,k}$ as the set of branch-adapted permutations,
\begin{equation}
\Pi_{l,k}
:=
\left\{
\pi\in S_k:
\pi(1)>\cdots>\pi(l),
\;
\pi(l+1)<\cdots<\pi(k)
\right\},
\end{equation}
namely the permutations that preserve the relative time ordering within the branch placed to the left of $O_P$ and within the branch placed to the right of $O_P$.
Equivalently, $\Pi_{l,k}$ only records which globally ordered times are assigned to the left and right branches; the actual branch ordering is then implemented explicitly by the product convention introduced below.
To keep the branch order explicit, we use the shorthand
\begin{equation}
\label{eq:supp_branch_products}
\prod_{a=1}^{l} X_a := X_1\cdots X_l,
\qquad
\prod_{b=l+1}^{k} Y_b := Y_{l+1}\cdots Y_k,
\end{equation}
with the convention that an empty product equals the identity.
The order-$k$ probe signal can then be written as
\begin{equation}
\label{eq:supp_dyson_raw}
\begin{aligned}
\langle O_P(T)\rangle_{\rho_P}
&=
\sum_{k\ge 0}(-i)^k
\sum_{\vec p}\sum_{l=0}^{k}\sum_{\pi\in \Pi_{l,k}}
\int d_>t_{[k]}\;
(-1)^l\,
\mathcal G_{O_P;\vec p}^{(k,l,\pi)}(t_{[k]})\;
\mathcal C_{\vec p,\pi}^{(k)}(t_{[k]}), \\
\mathcal G_{O_P;\vec p}^{(k,l,\pi)}(t_{[k]})
&:= 
\Tr_P\!\Bigg[
\prod_{a=1}^{l}
\tilde h_{p_a}(t_{\pi(a)})\,
O_P\,
\prod_{b=l+1}^{k}
\tilde h_{p_b}(t_{\pi(b)})\,
\rho_P
\Bigg],\\
\mathcal C_{\vec p,\pi}^{(k)}(t_{[k]})
&:= 
\Big\langle
\prod_{a=1}^{l}
M_{m(p_a)}(t_{\pi(a)})
\prod_{b=l+1}^{k}
M_{m(p_b)}(t_{\pi(b)})
\Big\rangle_{\rho_M}.
\end{aligned}
\end{equation}
Here both blocks are written in forward slot order.
The difference between the two branches is instead carried by the opposite monotonicity conventions in the definition of $\Pi_{l,k}$ above, together with the common product convention in Eq.~\eqref{eq:supp_branch_products}: on the left branch the permutation values decrease from left to right, whereas on the right branch they increase from left to right.
The permutation $\pi$ acts only on the ordered time labels $t_1,\ldots,t_k$; the probe-channel tuple $\vec p=(p_1,\ldots,p_k)$ remains fixed, and the actual matter operator in slot $p_j$ is always written as $M_{m(p_j)}$.
This is the same bookkeeping convention implemented later by $\widehat P_{\vec p}$ in the nested-braketor basis.
All probe-side dependence sits in $\mathcal G$, while the unknown matter information sits in the ordered $k$-point strings $\mathcal C$.

\paragraph{Attention.}
Under the present convention, the numerical ordering of the permutation values and the ordering of the physical times should not be conflated.
After expansion, the left string is
$\tilde h_{p_1}(t_{\pi(1)})\cdots \tilde h_{p_l}(t_{\pi(l)})$,
while the right string is
$\tilde h_{p_{l+1}}(t_{\pi(l+1)})\cdots \tilde h_{p_k}(t_{\pi(k)})$.
Thus the permutation values decrease from left to right on the left of $O_P$ and increase from left to right on the right of $O_P$.
However, because $t_1\ge\cdots\ge t_k$, the physical times $t_{\pi(\cdot)}$ increase from left to right on the left branch and decrease from left to right on the right branch.

The need for $\Pi_{l,k}$ is already visible at $k=2$.
The branch-resolved formula \eqref{eq:supp_branch_resolved} contains three populations $(r,s)=(2,0),(1,1),(0,2)$.
However, the mixed term $(r,s)=(1,1)$ still integrates over the full square $(t,t')\in[0,T]^2$ and must be split into the two ordered triangles $t>t'$ and $t'<t$.
Writing only the probe-side operator placement for clarity, one finds
\begin{equation}
\label{eq:supp_k2_mixed}
\begin{aligned}
\int_0^T dt\int_0^T dt'\;
\tilde h(t')\,O_P\,\tilde h(t)
&=
\int_{t'>t} dt\,dt'\;\tilde h(t')\,O_P\,\tilde h(t)
\;+\;
\int_{t<t'} dt\,dt'\;\tilde h(t')\,O_P\,\tilde h(t)\\
&=
\int d_>t_{[2]}
\Big[
\tilde h(t_1)\,O_P\,\tilde h(t_2)
\;+\;
\tilde h(t_2)\,O_P\,\tilde h(t_1)
\Big].
\end{aligned}
\end{equation}
These are precisely the two mixed operator placements in which one insertion lies on each branch,
namely $\tilde h(t_1)\,O_P\,\tilde h(t_2)$ and $\tilde h(t_2)\,O_P\,\tilde h(t_1)$.
Together with the purely right-branch placement $O_P\,\tilde h(t_1)\tilde h(t_2)$ and the purely left-branch placement $\tilde h(t_1)\tilde h(t_2)\,O_P$,
they generate the four order-$2$ configurations.
The role of $\Pi_{l,k}$ in Eq.~\eqref{eq:supp_dyson_raw} is to encode this refinement for general $k$ after one has already merged all branch variables into the single ordered list $t_1\ge \cdots\ge t_k$.

It is useful to encode $\pi$ by a binary branch-assignment string
\begin{equation}
\bm{\pi}_P=(\nu_1,\ldots,\nu_k)\in\{0,1\}^k,
\end{equation}
where $\nu_j=1$ means that the insertion at time $t_j$ is placed on the anti-time-ordered block to the left of $O_P$, whereas $\nu_j=0$ means that it is placed on the time-ordered block to the right.
The first bit $\nu_1$ produces a trivial two-fold redundancy because $O_P$ is absent in  $\mathcal{C}$, so only the reduced label
\begin{equation}
\bm{\pi}_M:=(\nu_2,\ldots,\nu_k)\in\{0,1\}^{k-1}
\end{equation}
matters for the sign structure of the braketor decomposition below.
Whenever a specific shuffle $\pi\in\Pi_{l,k}$ is under discussion, we use the same symbol $\bm{\pi}_M$ for its associated reduced binary label.

\subsection{Nested braketors as a complete basis}
\label{subsec:supp_braketor_basis}
For ordered times $t_1\ge \cdots \ge t_k$, define the nested-braketor correlators
\begin{equation}
\label{eq:supp_braketor_def}
\mathcal M_{\vec p}^{\vec\mu}(t_1,\ldots,t_k)
:=
\frac{1}{2^{k-1}}\,
\widehat P_{\vec p}
\Big(
\expect{
\llbracket\cdots\llbracket
M_{m(p_1)}(t_1),M_{m(p_2)}(t_2)\rrbracket_{\mu_1},
M_{m(p_3)}(t_3)\rrbracket_{\mu_2}\cdots,
M_{m(p_k)}(t_k)\rrbracket_{\mu_{k-1}}
}_{\rho_M}
\Big),
\end{equation}
where $\vec\mu=(\mu_1,\ldots,\mu_{k-1})\in\{0,1\}^{k-1}$ and $\widehat P_{\vec p}$ is the same bookkeeping convention used in the main text.
Its role is purely to keep the external probe-channel tuple $\vec p=(p_1,\ldots,p_k)$ in a fixed canonical order when different operator orderings are compared.
It is not a permutation over matter labels; the matter labels enter only through the fixed assignment inside $M_{m(p_j)}$.
The operator $\widehat P_{\vec p}$ does not change the braketor algebra itself; it only applies the channel-slot relabelling needed to write every term with the same ordered tuple $\vec p$.
At $k=2$, for instance,
\begin{equation}
\widehat P_{(p_1,p_2)}
\expect{M_{m(p_2)}(t_2)M_{m(p_1)}(t_1)}
=
\expect{M_{m(p_1)}(t_2)M_{m(p_2)}(t_1)}.
\end{equation}
This is why Eq.~\eqref{eq:supp_k2_inversion} can be written using a single label pair $(p_1,p_2)$ even though the two operator orderings differ.
The same bookkeeping persists at higher order: $\widehat P_{\vec p}$ keeps the channel labels fixed while the time-ordered braketor structure carries the actual operator-ordering information.
At fixed order $k$, there are therefore $2^{k-1}$ a priori distinct sectors.

The reason these objects form the correct basis is the elementary pair of identities
\begin{equation}
\label{eq:supp_xy_decomp}
XY=\frac12\acomm{X}{Y}+\frac12\comm{X}{Y},
\qquad
YX=\frac12\acomm{X}{Y}-\frac12\comm{X}{Y}.
\end{equation}
Every time a new operator is inserted into an already ordered string, there are exactly two possibilities:
it enters through an anti-commutator or through a commutator.
Iterating Eq.~\eqref{eq:supp_xy_decomp} therefore generates a binary tree of depth $(k-1)$, which is precisely the origin of the $2^{k-1}$ sectors.

For $k=2$ one finds the familiar pair
\begin{equation}
\label{eq:supp_k2_braketors}
\begin{aligned}
\mathcal M_{p_1,p_2}^{(0)}(t_1,t_2)
&=
\frac12\,\widehat P_{(p_1,p_2)}\expect{\acomm{M_{m(p_1)}(t_1)}{M_{m(p_2)}(t_2)}}\\
&=
\frac12\Big[
\expect{M_{m(p_1)}(t_1)M_{m(p_2)}(t_2)}
\;+\;
\expect{M_{m(p_1)}(t_2)M_{m(p_2)}(t_1)}
\Big],\\
\mathcal M_{p_1,p_2}^{(1)}(t_1,t_2)
&=
\frac12\,\widehat P_{(p_1,p_2)}\expect{\comm{M_{m(p_1)}(t_1)}{M_{m(p_2)}(t_2)}}\\
&=
\frac12\Big[
\expect{M_{m(p_1)}(t_1)M_{m(p_2)}(t_2)}
\;-\;
\expect{M_{m(p_1)}(t_2)M_{m(p_2)}(t_1)}
\Big].
\end{aligned}
\end{equation}
where the second operator string on each line is the $\widehat P_{(p_1,p_2)}$-relabelled version of
$\expect{M_{m(p_2)}(t_2)M_{m(p_1)}(t_1)}$.
They invert the two possible operator orderings:
\begin{equation}
\label{eq:supp_k2_inversion}
\begin{aligned}
\expect{M_{m(p_1)}(t_1)M_{m(p_2)}(t_2)}
&=
\mathcal M_{p_1,p_2}^{(0)}(t_1,t_2)+
\mathcal M_{p_1,p_2}^{(1)}(t_1,t_2),\\
\expect{M_{m(p_1)}(t_2)M_{m(p_2)}(t_1)}
&=
\mathcal M_{p_1,p_2}^{(0)}(t_1,t_2)-
\mathcal M_{p_1,p_2}^{(1)}(t_1,t_2).
\end{aligned}
\end{equation}
This is the lowest-order manifestation of the distinction emphasized in the main text:
system-only response is tied to the commutator sector, whereas probe readout also resolves the anti-commutator sector.

For $k=3$, the four sectors are
\begin{equation}
\begin{aligned}
\mathcal M^{(0,0)}_{p_1,p_2,p_3}
&=
\frac14\,\widehat P_{(p_1,p_2,p_3)}
\expect{\acomm{\acomm{M_{m(p_1)}(t_1)}{M_{m(p_2)}(t_2)}}{M_{m(p_3)}(t_3)}},\\
\mathcal M^{(1,0)}_{p_1,p_2,p_3}
&=
\frac14\,\widehat P_{(p_1,p_2,p_3)}
\expect{\acomm{\comm{M_{m(p_1)}(t_1)}{M_{m(p_2)}(t_2)}}{M_{m(p_3)}(t_3)}},\\
\mathcal M^{(0,1)}_{p_1,p_2,p_3}
&=
\frac14\,\widehat P_{(p_1,p_2,p_3)}
\expect{\comm{\acomm{M_{m(p_1)}(t_1)}{M_{m(p_2)}(t_2)}}{M_{m(p_3)}(t_3)}},\\
\mathcal M^{(1,1)}_{p_1,p_2,p_3}
&=
\frac14\,\widehat P_{(p_1,p_2,p_3)}
\expect{\comm{\comm{M_{m(p_1)}(t_1)}{M_{m(p_2)}(t_2)}}{M_{m(p_3)}(t_3)}}.
\end{aligned}
\end{equation}
The ordered string with descending times is simply the sum of all four sectors ($t_1\geq t_2\geq t_3$),
\begin{equation}
\label{eq:supp_k3_direct}
\expect{M_{m(p_1)}(t_1)M_{m(p_2)}(t_2)M_{m(p_3)}(t_3)}
=
\sum_{\mu_1,\mu_2\in\{0,1\}}
\mathcal M^{(\mu_1,\mu_2)}_{p_1,p_2,p_3}(t_1,t_2,t_3),
\end{equation}
whereas permuted strings differ only by signs.
For example,
\begin{equation}
\label{eq:supp_k3_permuted}
\begin{aligned}
\expect{M_{m(p_1)}(t_1)M_{m(p_2)}(t_3)M_{m(p_3)}(t_2)}
&=
\mathcal M^{(0,0)}_{p_1,p_2,p_3}(t_1,t_2,t_3)
-\mathcal M^{(1,0)}_{p_1,p_2,p_3}(t_1,t_2,t_3)
\\
&\quad +\mathcal M^{(0,1)}_{p_1,p_2,p_3}(t_1,t_2,t_3)
-\mathcal M^{(1,1)}_{p_1,p_2,p_3}(t_1,t_2,t_3),
\end{aligned}
\end{equation}
where the shared arguments $(p_1,p_2,p_3;t_1,t_2,t_3)$ are understood.
This sign pattern is the concrete low-order example of the general rule encoded by $\bm{\pi}_M$.

The general statement is the following: for each fixed branch-assignment class $\bm{\pi}_M$,
\begin{equation}
\label{eq:supp_general_string_expansion}
\mathcal C_{\vec p,\pi}^{(k)}(t_{[k]})
=
\sum_{\vec\mu\in\{0,1\}^{k-1}}
(-1)^{\vec\mu\cdot \bm{\pi}_M}\,
\mathcal M_{\vec p}^{\vec\mu}(t_{[k]}).
\end{equation}
Equation~\eqref{eq:supp_general_string_expansion} is proved recursively by repeatedly applying Eq.~\eqref{eq:supp_xy_decomp}.
The only effect of moving one insertion from one branch to the other is to flip the sign of the commutator choice at the corresponding nesting step; anti-commutator choices are unaffected.
That is why the coefficient is exactly the binary phase $(-1)^{\vec\mu\cdot \bm{\pi}_M}$.
For example, if $\pi$ is such that
\begin{equation*}
\mathcal C_{\vec p,\pi}^{(k)}(t_{[k]})
=
\expect{M_{m(p_1)}(t_1)\cdots M_{m(p_k)}(t_k)},
\end{equation*}
then every insertion lies on the time-ordered block to the right and hence $\bm{\pi}_M=(0,\ldots,0)$.
Equation~\eqref{eq:supp_general_string_expansion} then reduces to the direct ordered-string sum with all positive signs.

\subsection{Control tensor and universal probe-readout formula}
\label{subsec:supp_control_tensor}
Substituting Eq.~\eqref{eq:supp_general_string_expansion} into the ordered Dyson series \eqref{eq:supp_dyson_raw} yields the universal expansion used in the main text:
\begin{equation}
\label{eq:supp_universal_readout}
\langle O_P(T)\rangle_{\rho_P}
=
\sum_{k\ge 0}\sum_{\vec p}\sum_{\vec\mu\in\{0,1\}^{k-1}}
(-i)^k
\int d_>t_{[k]}\;
\mathcal F_{O_P;\vec p;t_{[k]}}^{(k;\vec\mu)}[H_P]\;
\mathcal M_{\vec p}^{\vec\mu}(t_{[k]}),
\end{equation}
with the explicit control tensor
\begin{equation}
\label{eq:supp_control_tensor_def}
\mathcal F_{O_P;\vec p;t_{[k]}}^{(k;\vec\mu)}[H_P]
:=
\sum_{l=0}^{k}\sum_{\pi\in\Pi_{l,k}}
(-1)^{l+\vec\mu\cdot \bm{\pi}_M}\,
\mathcal G_{O_P;\vec p}^{(k,l,\pi)}(t_{[k]}).
\end{equation}
Equation~\eqref{eq:supp_universal_readout} is the precise mathematical form of the ``learning interface'' described in the main text.
All dependence on experimental design is isolated in $\mathcal F$, while all dependence on the unknown many-body system is isolated in the braketor correlators $\mathcal M_{\vec p}^{\vec\mu}$.
This is exactly what makes the framework operational:
the laboratory learns matter by engineering probe-side tensors that isolate desired sectors of the matter algebra.

\subsection{Minimal comparison with system-only response}
\label{subsec:supp_response_comparison}

A classical-probe or system-only protocol should not be read as a protocol without quantum matter.
Rather, it is the limit in which the external drive carries no independently coherent probe Hilbert space whose reduced readout is measured at the end.
The effective description is therefore a driven target-system response problem.

For comparison, consider a conventional system-only protocol with Hamiltonian
\begin{equation}
H(t_j)=H_0-\lambda(t_j)M_{m_j}(t_j),
\end{equation}
and some system observable $O_M$ measured at time $T$.
Here $m_j$ denotes the matter channel selected by the classical perturbation at the insertion time $t_j$.
Equivalently, the experimental schedule directly pairs each response-theory time slot $t_j$ with the driven matter operator $M_{m_j}$; in the single-channel case all $m_j$ are the same.
In the interaction picture with respect to $H_0$, the driven expectation value obeys
\begin{equation}
\label{eq:supp_response_series}
\langle O_M(T)\rangle_{\lambda}
=
\sum_{k\ge 0}
i^k
\int d_>t_{[k]}\;
\lambda(t_1)\cdots\lambda(t_k)\;
\expect{\comm{\cdots\comm{\comm{O_M(T)}{M_{m_1}(t_1)}}{M_{m_2}(t_2)}}{M_{m_k}(t_k)}}.
\end{equation}
In Eq.~\eqref{eq:supp_response_series}, the matter label is displayed directly under each operator as $m_j$, paired with the corresponding response-theory time $t_j$.
The explicit observable $O_M(T)$ is the standard response-theory convention.
For comparison with the all-matter notation used in the main-text Methods table, take the $(k-1)$th-order response, choose $O_M(T)=M_{m_1}(t_1)$ after renaming the observation time $T$ as $t_1$, and denote the perturbation insertions by $M_{m_2}(t_2),\ldots,M_{m_k}(t_k)$.
The same sector is then written as
\begin{equation}
\expect{\comm{\cdots\comm{\comm{M_{m_1}(t_1)}{M_{m_2}(t_2)}}{M_{m_3}(t_3)}}{M_{m_k}(t_k)}},
\end{equation}
with no change in physical content.
This is the system-only analogue of the channel tuple used above, but with a different bookkeeping role:
there $p_j$ labels a probe coupling channel, whose assigned matter operator is written explicitly as $M_{m(p_j)}$ once the map from probe channel to matter channel is fixed.
At each order there is only one sector:
the fully retarded nested commutator.
No forward--backward reduction occurs because the measured quantum system is the target itself rather than a traced-down probe subsystem.
Equation~\eqref{eq:supp_response_series} is therefore the closed-system counterpart of Eq.~\eqref{eq:supp_universal_readout}, and the difference between them is the exact algebraic content of quantum probe advantage.

\section{Formal statements and proofs}
\label{sec:supp_proofs}
This section restates the main structural results and proves them using the machinery of Sec.~\ref{sec:supp_general_algebra}.
The goal is to make two statements precise.
First, reduced probe dynamics does not merely repackage the single retarded response function:
at order $k$ it naturally produces a full binary family of operator sectors.
Second, once several probe settings are used, learning those sectors is an ordinary linear inverse problem.

There is one point of bookkeeping that is easy to miss.
After the state, channel tuple, and times are fixed, the correlators below are just complex numbers.
The dimension statements in this section are therefore not claims that a particular list of numbers is linearly independent.
They refer to the formal coordinate space of possible order-$k$ operator sectors before extra state-dependent symmetries, equilibrium relations, equal-time degeneracies, or special operator identities are imposed.
The measured numbers are coordinates in that space.

Fix a Dyson order $k$, a channel tuple $\vec p=(p_1,\ldots,p_k)$, and ordered times $t_1\ge\cdots\ge t_k$.
There are two useful coordinate systems for the same order-$k$ information.
The first is the braketor-sector coordinate system, whose entries are nested commutator/anti-commutator correlators.
We write its coordinate vector as
\begin{equation}
\label{eq:supp_sector_vector}
\mathbf M^{(k)}_{\vec p}(t_{[k]})
:=
\Big(
\mathcal M_{\vec p}^{\vec\mu}(t_{[k]})
\Big)_{\vec\mu\in\{0,1\}^{k-1}}
\in \mathbb C^{2^{k-1}}
\end{equation}
where the binary label $\vec\mu$ records the braketor choice at each nesting step:
$\mu_j=0$ denotes an anti-commutator and $\mu_j=1$ denotes a commutator.

The second coordinate system is closer to the raw Dyson expansion.
It is labelled by the reduced forward/backward branch assignment $\bm\beta\in\{0,1\}^{k-1}$.
For each such branch label, define the ordered branch-string correlator
\begin{equation}
\label{eq:supp_branch_vector}
\mathcal C^{(k)}_{\vec p,\bm\beta}(t_{[k]})
:=
\mathcal C^{(k)}_{\vec p,\pi}(t_{[k]})
\qquad
\text{for any }\pi\text{ with associated reduced label }\bm{\pi}_M=\bm\beta,
\end{equation}
which is well defined because Eq.~\eqref{eq:supp_general_string_expansion} depends on $\pi$ only through $\bm{\pi}_M$.
Collect these components into
\begin{equation}
\mathbf C^{(k)}_{\vec p}(t_{[k]})
:=
\Big(
\mathcal C^{(k)}_{\vec p,\bm\beta}(t_{[k]})
\Big)_{\bm\beta\in\{0,1\}^{k-1}}
\in \mathbb C^{2^{k-1}}.
\end{equation}
Thus $\mathbf C$ is the branch-string description and $\mathbf M$ is the sector description.
The next lemma says that these are not two different physical assumptions; they are two bases related by a discrete Fourier transform on the binary cube.

\begin{lemma}[Walsh--Hadamard relation between branch strings and sectors]
\label{lem:supp_hadamard}
For each fixed $k,\vec p,t_{[k]}$, the vectors $\mathbf C^{(k)}_{\vec p}$ and $\mathbf M^{(k)}_{\vec p}$ are related by
\begin{equation}
\label{eq:supp_hadamard_forward}
\mathbf C^{(k)}_{\vec p}
=
H_{k-1}\,
\mathbf M^{(k)}_{\vec p},
\qquad
(H_{k-1})_{\bm\beta,\vec\mu}:=
(-1)^{\bm\beta\cdot \vec\mu},
\end{equation}
where $H_{k-1}$ is the $(2^{k-1}\times 2^{k-1})$ Walsh--Hadamard matrix.
Rows of $H_{k-1}$ are labelled by branch strings $\bm\beta$, columns by braketor sectors $\vec\mu$.
The entry $(-1)^{\bm\beta\cdot\vec\mu}$ is the sign acquired by the sector $\vec\mu$ when the ordered operator string is moved into the branch pattern $\bm\beta$.
In particular,
\begin{equation}
\label{eq:supp_hadamard_inverse}
\mathbf M^{(k)}_{\vec p}
=
2^{-(k-1)}\,H_{k-1}\,
\mathbf C^{(k)}_{\vec p}.
\end{equation}
Thus either list determines the other.
No measurement assumption enters here; this is only the algebraic change of coordinates between ordered strings and nested braketors.
\end{lemma}

\begin{proof}
Equation~\eqref{eq:supp_general_string_expansion} says exactly that
\begin{equation}
\mathcal C^{(k)}_{\vec p,\bm\beta}(t_{[k]})
=
\sum_{\vec\mu\in\{0,1\}^{k-1}}
(-1)^{\bm\beta\cdot\vec\mu}\,
\mathcal M_{\vec p}^{\vec\mu}(t_{[k]}),
\end{equation}
which is the component form of Eq.~\eqref{eq:supp_hadamard_forward}.
To invert it, use the standard Walsh--Hadamard orthogonality relation
\begin{equation}
\sum_{\bm\beta\in\{0,1\}^{k-1}}
(-1)^{\bm\beta\cdot(\vec\mu+\vec\nu)}
=
2^{k-1}\delta_{\vec\mu,\vec\nu},
\end{equation}
where the addition in the exponent is modulo two.
This relation says that two distinct binary sign characters have zero overlap when summed over all branch labels, while a character has norm squared $2^{k-1}$.
Equivalently,
\begin{equation}
H_{k-1}^2=2^{k-1}\,\mathbb I.
\end{equation}
Multiplying Eq.~\eqref{eq:supp_hadamard_forward} by $2^{-(k-1)}H_{k-1}$ gives Eq.~\eqref{eq:supp_hadamard_inverse}.
\end{proof}

\begin{proposition}[The order-$k$ probe algebra is genuinely $2^{k-1}$-dimensional]
\label{prop:supp_full_sector_space}
Fix $k,\vec p,t_{[k]}$ and let
\begin{equation}
\mathsf V_k(\vec p,t_{[k]})
:=
\Big\{
(x_{\bm\beta})_{\bm\beta\in\{0,1\}^{k-1}}
:x_{\bm\beta}\in\mathbb C
\Big\}
\cong \mathbb C^{2^{k-1}}
\end{equation}
be the formal branch-coordinate space for order-$k$ matter strings.
The branch-string coordinates $\mathbf C^{(k)}_{\vec p}$ and the braketor-sector coordinates $\mathbf M^{(k)}_{\vec p}$ are related by an invertible change of coordinates on this space.
Therefore
\begin{equation}
\dim \mathsf V_k(\vec p,t_{[k]})=2^{k-1}.
\end{equation}
Equivalently, before any extra state-dependent symmetry or equilibrium constraint is imposed, the order-$k$ reduced probe signal depends on matter through a full $2^{k-1}$-component sector vector.
\end{proposition}

\begin{proof}
Lemma~\ref{lem:supp_hadamard} provides an explicit invertible change of coordinates between the branch-string vector $\mathbf C^{(k)}_{\vec p}$ and the nested-braketor vector $\mathbf M^{(k)}_{\vec p}$.
Indeed,
\begin{equation}
\mathbf C^{(k)}_{\vec p}=H_{k-1}\mathbf M^{(k)}_{\vec p},
\qquad
\mathbf M^{(k)}_{\vec p}=2^{-(k-1)}H_{k-1}\mathbf C^{(k)}_{\vec p},
\end{equation}
with $H_{k-1}$ invertible.
Hence passing from branch strings to braketor sectors is a change of coordinates, not a projection.
Since the branch-coordinate space has one coordinate for each reduced branch label $\bm\beta\in\{0,1\}^{k-1}$, it has dimension $2^{k-1}$, and the braketor-sector coordinates give an equivalent description of the same space.

The final sentence is just the operational reading of this coordinate-space statement.
Equation~\eqref{eq:supp_universal_readout} shows that every order-$k$ probe contribution is a linear functional of $\mathbf M^{(k)}_{\vec p}$.
Because $\mathbf M^{(k)}_{\vec p}$ supplies coordinates on the full probe-accessible order-$k$ matter space, there is no universal reduction to a smaller sector list unless one imposes additional constraints on the target state or dynamics.
Such constraints can occur in special states or limits, but they are not part of the general probe algebra.
\end{proof}

We now pass from the algebraic coordinate space to experimentally generated data.
Consider a family of probe settings indexed by $\alpha$.
Each setting $\alpha$ means a complete probe-side choice:
probe control history, initial probe state, and final probe observable.
For the purposes of the rank argument, we focus on a fixed ordered time sample $t_{[k]}$ and a fixed channel tuple $\vec p$.
Equivalently, one may regard the equations below as the fixed-time kernel inside the full time integral, or as the object obtained after window localization or deconvolution.
This is why the following vectors carry an explicit $t_{[k]}$ dependence even though the experimentally recorded signal is ultimately integrated over time.
For fixed order $k,\vec p,t_{[k]}$, define the order-$k$ measured data by
\begin{equation}
\mathbf d^{(k)}_{\vec p}(t_{[k]})
:=
\big(d^{(k)}_{\alpha;\vec p}(t_{[k]})\big)_\alpha,
\end{equation}
where each component is the order-$k$ contribution to the corresponding probe signal.
The index $\alpha$ labels rows of the eventual linear system:
changing $\alpha$ changes only the probe-side control tensor, not the unknown matter correlators.

Grouping Eq.~\eqref{eq:supp_dyson_raw} by the reduced branch label $\bm\beta=\bm{\pi}_M$ gives
\begin{equation}
\label{eq:supp_data_branch}
d^{(k)}_{\alpha;\vec p}
=
\sum_{\bm\beta\in\{0,1\}^{k-1}}
G^{(k)}_{\alpha;\vec p,\bm\beta}\,
\mathcal C^{(k)}_{\vec p,\bm\beta},
\end{equation}
with
\begin{equation}
\label{eq:supp_G_matrix_def}
G^{(k)}_{\alpha;\vec p,\bm\beta}
:=
\sum_{\substack{l,\pi\\ \bm{\pi}_M=\bm\beta}}
(-1)^l\,
\mathcal G^{(k,l,\pi)}_{\alpha;O_P,\vec p}(t_{[k]}).
\end{equation}
Here $G^{(k)}_{\alpha;\vec p,\bm\beta}$ is known once the probe setting is chosen.
It collects all forward/backward Dyson terms that lead to the same matter branch string $\bm\beta$.
Thus $G^{(k)}_{\vec p}$ is the control matrix written in the branch-string coordinate system:
its rows are experimental settings $\alpha$, and its columns are branch labels $\bm\beta$.

Equivalently, using Lemma~\ref{lem:supp_hadamard},
\begin{equation}
\label{eq:supp_data_sector}
d^{(k)}_{\alpha;\vec p}
=
\sum_{\vec\mu\in\{0,1\}^{k-1}}
F^{(k)}_{\alpha;\vec p,\vec\mu}\,
\mathcal M_{\vec p}^{\vec\mu},
\qquad
F^{(k)}_{\alpha;\vec p,\vec\mu}
=
\sum_{\bm\beta}
G^{(k)}_{\alpha;\vec p,\bm\beta}\,
(-1)^{\bm\beta\cdot\vec\mu}.
\end{equation}
The coefficients $F^{(k)}_{\alpha;\vec p,\vec\mu}$ are the same probe settings rewritten in the braketor-sector coordinate system.
They say how strongly setting $\alpha$ weights each nested commutator/anti-commutator sector $\vec\mu$.
In matrix form,
\begin{equation}
\label{eq:supp_F_equals_GH}
\mathbf d^{(k)}_{\vec p}
=
G^{(k)}_{\vec p}\,
\mathbf C^{(k)}_{\vec p}
=
F^{(k)}_{\vec p}\,
\mathbf M^{(k)}_{\vec p},
\qquad
F^{(k)}_{\vec p}=G^{(k)}_{\vec p}\,H_{k-1}.
\end{equation}
This is the central linear-algebra object of the theorem.
The unknown vector is $\mathbf M^{(k)}_{\vec p}$; the measured vector is $\mathbf d^{(k)}_{\vec p}$; and the matrix $F^{(k)}_{\vec p}$ is engineered by probe design.

\begin{lemma}[Equivalent operational resolution criteria]
\label{lem:supp_rank_equiv}
For any family of probe settings at fixed $k,\vec p,t_{[k]}$,
\begin{equation}
\rank F^{(k)}_{\vec p}=\rank G^{(k)}_{\vec p}.
\end{equation}
Therefore the probe family resolves the full braketor sector space if and only if it resolves the full branch-string space.
\end{lemma}

\begin{proof}
Equation~\eqref{eq:supp_F_equals_GH} gives
\begin{equation}
F^{(k)}_{\vec p}=G^{(k)}_{\vec p}H_{k-1}.
\end{equation}
By Lemma~\ref{lem:supp_hadamard}, $H_{k-1}$ is invertible.
Right multiplication by an invertible matrix only changes the coordinate basis of the columns.
It cannot create or remove independent directions in the column space, so
\begin{equation}
\rank F^{(k)}_{\vec p}=\rank G^{(k)}_{\vec p}.
\end{equation}
Operationally, this means that it does not matter whether one designs probe settings to distinguish branch strings first and then converts them to braketors, or directly designs them to distinguish braketor sectors.
The rank criterion is the same.
\end{proof}

\begin{theorem}[Quantum probe advantage: algebraic accessibility and operational reconstruction]
\label{thm:supp_probe_adv}
Fix Dyson order $k$.
At this order, reduced probe readout is a linear inverse problem on the full $2^{k-1}$-dimensional sector space of Proposition~\ref{prop:supp_full_sector_space}, whereas system-only response is a linear inverse problem on the single fully retarded sector.
There are two separate meanings of this statement.
The algebraic accessibility statement says what the probe signal can depend on in principle.
The reconstruction statement says when a chosen collection of settings contains enough independent rows to recover the whole vector.
More explicitly, for any chosen family of probe settings one has
\begin{equation}
\mathbf d^{(k)}_{\vec p}
=
F^{(k)}_{\vec p}\,
\mathbf M^{(k)}_{\vec p},
\end{equation}
where $\mathbf d^{(k)}_{\vec p}$ is the data vector whose component for probe setting $\alpha$ is the measured order-$k$ contribution $d^{(k)}_{\alpha;\vec p}$ at the fixed times $t_{[k]}$, and $\mathbf M^{(k)}_{\vec p}\in\mathbb C^{2^{k-1}}$ is the unknown sector vector.
If, for the target order-$k$ channel tuple and times of interest, the associated control matrix satisfies
\begin{equation}
\rank F^{(k)}_{\vec p}=2^{k-1}
\qquad
\text{equivalently}\qquad
\rank G^{(k)}_{\vec p}=2^{k-1},
\end{equation}
then the full order-$k$ braketor sector vector $\mathbf M^{(k)}_{\vec p}$ is uniquely recoverable from probe data.
If the rank is smaller, the unresolved sectors are precisely the directions in the null space of $F^{(k)}_{\vec p}$.
\end{theorem}

\begin{proof}
The first claim is the combination of Proposition~\ref{prop:supp_full_sector_space} with Eq.~\eqref{eq:supp_F_equals_GH}.
The proposition identifies the relevant matter coordinate space as the full $2^{k-1}$-dimensional sector space, and Eq.~\eqref{eq:supp_F_equals_GH} shows that a chosen probe family acts on that space through the linear map $F^{(k)}_{\vec p}$.
Thus the order-$k$ probe readout is not a single response coefficient; it is a set of linear measurements of a $2^{k-1}$-component vector.

If $\rank F^{(k)}_{\vec p}=2^{k-1}$, then $F^{(k)}_{\vec p}$ has full column rank on the $2^{k-1}$-dimensional sector space.
Equivalently, the chosen probe settings provide enough independent linear combinations to separate all sector coordinates.
Hence the linear map
\begin{equation}
\mathbf M^{(k)}_{\vec p}\mapsto \mathbf d^{(k)}_{\vec p}
\end{equation}
is injective, so $\mathbf M^{(k)}_{\vec p}$ is uniquely determined by the measured data.

By Lemma~\ref{lem:supp_rank_equiv}, the same criterion may equivalently be checked in the branch-string basis.
Operationally, this means that the probe family can resolve the full $2^{k-1}$-dimensional sector space exactly when its probe-side control tensors span a $2^{k-1}$-dimensional subspace.

For system-only response, no analogous enlargement exists.
Equation~\eqref{eq:supp_response_series} shows that the order-$k$ signal depends only on
\begin{equation}
\expect{\comm{\cdots\comm{\comm{O_M(T)}{M_{m_1}(t_1)}}{M_{m_2}(t_2)}}{M_{m_k}(t_k)}},
\end{equation}
namely the single fully retarded nested-commutator sector.
Therefore system-only response is one-dimensional in sector space, whereas reduced probe readout is formulated on the full $2^{k-1}$-dimensional space and a suitably chosen probe family can resolve that full space.
This is the precise sense in which the advantage is algebraic rather than merely a matter of signal sensitivity.
\end{proof}

\begin{remark}
\label{rem:supp_can_learn}
The theorem makes precise the phrase ``a probe can learn $2^{k-1}$ sectors'' used in the main text.
The statement has two parts.
The algebraic enlargement is unconditional: reduced probe dynamics lives on the full $2^{k-1}$-dimensional sector space.
The reconstruction claim is conditional: one must still choose a tomographically complete family of probe settings so that the corresponding control matrix has full column rank.
Thus ``can learn'' does not mean that one fixed probe setting isolates all sectors automatically; it means that the probe framework turns the problem into a full-space inverse problem that can, in principle, be made injective by control design.
\end{remark}

\begin{corollary}[Probe count tracks correlator order]
\label{cor:supp_probe_count}
Suppose the target quantity depends on matter correlators up to order $K$.
Then $K$ coherently controlled probe degrees of freedom are sufficient in the worst case.
In particular, the required probe overhead scales as $\mathcal O(K)$, independent of the size of the many-body system.
\end{corollary}

\begin{proof}
Consider a target $K$-point sector
\begin{equation}
\mathcal M_{p_1,\ldots,p_K}^{\vec\mu}(t_1,\ldots,t_K),
\end{equation}
specified by probe channels $(p_1,\ldots,p_K)$, the fixed assignment $p_j\mapsto m(p_j)$, and times $(t_1,\ldots,t_K)$.
By Eq.~\eqref{eq:supp_universal_readout}, realizing this sector requires $K$ operator insertions along the probe--matter interaction history.
An explicit worst-case construction is therefore to use $K$ probe degrees of freedom, one for each insertion slot, and couple probe $p_j$ to $M_{m(p_j)}$ at time $t_j$.
This already provides a sufficient implementation for every order-$K$ sector.

This bound is generally not tight.
The tuple $(p_1,\ldots,p_K)$ specifies the addressed interaction slots needed by the target $K$-point sector.
In the direct-addressing convention used in the main text, these slots are carried by distinct physical probes, so channel count and probe count coincide.
The present corollary is a worst-case sufficiency statement: it also allows hardware in which a physical probe can be reused across non-overlapping time windows.
For example, in the single-probe dephasing setting at $K=2$ the same channel addresses the same matter operator at both insertions, so one probe already accesses the two-point sector.
More generally, if the target insertions occur at distinct times and the interface can be retuned between windows, the same physical probe may implement different addressed slots in successive windows.
The genuine worst case is when distinct addressed operators must be coupled simultaneously, for instance when $t_1=\cdots=t_K$.
Then the insertion slots cannot be recycled into fewer physical probes, and $K$ probes are operationally sufficient.

The important point is that this worst-case resource count depends only on the correlator order.
The size of the many-body system affects which matter operators are addressed and how many experimental settings are required, but it does not force the probe register to scale with the Hilbert-space dimension of the target.
Hence the probe overhead is $\mathcal O(K)$.
\end{proof}

\begin{corollary}[Experimental-setting count for interacting entropy learning]
\label{cor:supp_setting_count}
Assume the entropy of a subsystem $A$ is reconstructed perturbatively up to order $n$ from correlators of a local operator basis on $A$ whose cardinality scales as $\mathcal O(|A|)$.
Then, for fixed perturbative order $n$, at most $\mathcal O(|A|^n)$ distinct experimental settings are needed, and each setting can be implemented with $\mathcal O(n)$ coherently controlled probes.
\end{corollary}

\begin{proof}
Under the stated locality assumption, the perturbative expansion up to order $n$ involves operator tuples of lengths $m=1,\ldots,n$ drawn from a basis of size $\mathcal O(|A|)$.
Hence the number of candidate tuples is bounded by
\begin{equation}
\sum_{m=1}^{n}\mathcal O(|A|^m)
=
\mathcal O(|A|^n).
\end{equation}
Each tuple labels one family of braketor sectors to be learned.
For a tuple of length $m\le n$, resolving those sectors requires a tomographically complete family of probe settings on the corresponding $m$-point sector space.
By Theorem~\ref{thm:supp_probe_adv} and Remark~\ref{rem:supp_can_learn}, this tomography overhead depends only on the sector dimension at order $m$, hence only on $m$ and not on $|A|$.
For fixed perturbative order $n$, it is therefore an overall $n$-dependent constant that does not affect the asymptotic $\mathcal O(|A|^n)$ scaling.

By Corollary~\ref{cor:supp_probe_count}, each such setting can be implemented with at most $\mathcal O(m)\subseteq \mathcal O(n)$ coherently controlled probes.
Therefore the total number of experimental settings scales as $\mathcal O(|A|^n)$ for fixed $n$, with probe count $\mathcal O(n)$ per setting.
\end{proof}

\section{Constructive control design for correlator-selective learning}
\label{sec:supp_constructive_design}
The Methods section mentioned both always-on and windowed reconstruction regimes, but the constructive design in this section focuses on the windowed case.
The always-on or broad-window setting, illustrated in Fig.~\ref{fig:supp_window_pm}(a), is different:
the interaction remains active across the target interval, and the probe signal contains a genuine time convolution over possible insertion times and channel placements.
Such data can still be deconvolved, but the kernels are less sparse, as in the DISCO framework.

We now turn to the windowed construction shown in Fig.~\ref{fig:supp_window_pm}(b,c).
Here the interaction is turned on only in short windows around chosen target times, so individual ordered correlator samples can be isolated directly.
The design principle is to localize the probe--matter coupling near the target times and then vary the setting-dependent probe--matter connection so that different ordered strings appear with distinguishable weights.

Fix a probe-channel tuple $\vec p=(p_1,\ldots,p_k)$, a target matter-operator tuple $\vec m=(m_1,\ldots,m_k)$, and target times
\begin{equation}
\bar\tau=(\tau_1,\ldots,\tau_k),
\qquad
\tau_1\geq\tau_2\geq\cdots\geq\tau_k.
\end{equation}
Our ultimate targets are the generalized braketor sectors
\(
\mathcal M_{\vec p}^{(\vec{\mu})}(\bar\tau)
\),
with $\vec{\mu}\in\{0,1\}^{k-1}$ as defined earlier in the Supplementary Material.
Here $\vec p$ records the probe slots, while $\vec m$ records the matter operators to be placed into those slots by the setting-dependent control.
To explain how these quantities are accessed experimentally, it is useful to first write the idealized short-window limit.
Let
\begin{equation}
\Pi_k:=\bigsqcup_{l=0}^{k}\Pi_{l,k},
\qquad
|\Pi_k|=2^k,
\end{equation}
denote the full family of branch-assignment permutations appearing in the ordered Dyson expansion.
In this constructive setting, unlike the fixed-channel convention used in the main text, the channel-to-matter map may be changed between experimental settings.
For each fixed setting $\eta\in\Pi_k$, choose the controlled probe-side operators by $\tilde h_{p_j}^{[\eta]}(t)=\delta(t-\tau_{\eta(j)})\tilde h_{p_j}(t)$ for the selected channels and $\tilde h_p^{[\eta]}(t)=0$ for unused channels, and choose a run-specific but time-independent matter assignment $m_\eta(p)$ so that $m_\eta(p_j)=m_j$.
This means probe slot $p_j$ is coupled to $M_{m_j}$ during its active window centred at $\tau_{\eta(j)}$.
The experimentally chosen window pattern (cf Eq.~\eqref{eq:toggled_PM_H}) is then
\begin{equation}
\label{eq:supp_constructive_HI}
H_I^{[\eta]}(t)
=
\sum_{p\in P}
\tilde h_p^{[\eta]}(t)\otimes M_{m_\eta(p)}(t).
\end{equation}
Here $\eta$ labels the whole experimental setting: it changes the activation-window schedule and, when needed, the fixed channel-to-matter assignment used for that run.
Finite-width windows are obtained by replacing the delta envelope with $f_\sigma(t-\tau_{\eta(j)})$, with the present formula recovered as $\sigma\to0$.
Notice that we have already used the fact that $k$ probes to learn $k$-point correlators are sufficient.
In the narrow-window limit, the order-$k$ contribution to probe readout is dominated by operator insertions at the selected times.
Although the desired outputs are the sectors $\mathcal M_{\vec p}^{(\vec{\mu})}(\bar\tau)$, the window-selection mechanism acts one level earlier in the expansion:
the Dyson series is organized by explicit branch assignments and therefore isolates ordered strings $\mathcal C_{\vec p,\bm\beta}^{(k)}(\bar\tau)$ directly.
The target sectors $\mathcal M$ are then obtained from the component relation between ordered strings and braketor sectors, Eq.~\eqref{eq:supp_general_string_expansion}, after the required set of $\mathcal C$ strings has been learned.
Operationally, the window pattern selects which ordered sample enters the signal, while the passage from $\mathcal C$ to the target $\mathcal M$ is a subsequent algebraic recombination.

Substituting Eq.~\eqref{eq:supp_constructive_HI} into Eq.~\eqref{eq:supp_dyson_raw}, every contribution carries a product of delta functions that enforces the chosen timing assignment.
Because the integration domain already satisfies $t_1\ge\cdots\ge t_k$ and the target centres satisfy $\tau_1\ge\cdots\ge\tau_k$, the surviving terms are not characterised by the full raw permutation $\pi$ itself, but by its reduced matter label $\bm{\pi}_M$.
Likewise, the experimentally chosen window pattern $\eta\in\Pi_k$ has an associated reduced label $\bm{\eta}_M$.
Only terms with
\begin{equation}
\bm{\pi}_M=\bm{\eta}_M
\end{equation}
survive, while the others vanish.
Thus the original sum over the $2^k$ raw branch assignments collapses onto the single reduced equivalence class selected by $\bm{\eta}_M$.
Since $\mathcal C_{\vec p,\pi}^{(k)}$ depends only on the reduced label, the common surviving matter object is simply
\begin{equation}
\mathcal C_{\vec p,\bm{\eta}_M}^{(k)}(\bar\tau).
\end{equation}
Correspondingly, for a probe setting $\alpha$ the order-$k$ readout reduces schematically to
\begin{equation}
\label{eq:supp_constructive_branch}
\langle O_P(T)\rangle^{(k)}_{\alpha,[\eta]}
\approx
(-i)^k\,
G^{(k)}_{\alpha;\vec p,\bm{\eta}_M}(\bar\tau)\,
\mathcal C_{\vec p,\bm{\eta}_M}^{(k)}(\bar\tau),
\end{equation}
up to corrections that vanish as the window width is reduced.
Here
\(
G^{(k)}_{\alpha;\vec p,\bm{\eta}_M}(\bar\tau)
\)
is the same branch-string control-matrix element defined in Eq.~\eqref{eq:supp_G_matrix_def}, evaluated for the windowed setting $\eta$ and the chosen probe setting $\alpha$.
Equivalently, it is obtained by summing the elementary probe-side traces $\mathcal G$ over the raw branch assignments whose reduced label is $\bm{\eta}_M$, with the window envelopes evaluated near $\bar\tau$.

The two labels serve different purposes: $\eta$ routes the selected probe slots to target times and matter operators, thereby choosing the branch-string column, whereas $\alpha$ changes only the probe preparation, control, and readout, thereby supplying rows of the calibrated linear system.
Thus $G$ is known once the experimental controls are chosen; it is not an additional matter correlator.
In the ideal delta-window limit all other reduced labels have zero coefficient; for finite windows they produce controlled leakage terms that vanish as the window width is reduced.
The important point is that the window pattern $\eta$ does not merely localise the integrals near $\bar\tau$.
It selects the reduced ordered string labelled by $\bm{\eta}_M$.

The reconstruction logic is therefore straightforward.
For fixed $(\vec p,\vec m,\bar\tau)$, one first learns the ordered branch strings $\mathcal C_{\vec p,\bm\beta}^{(k)}(\bar\tau)$ by repeating the experiment with different choices of $\eta$.
At the raw level there are $2^k$ window assignments, in one-to-one correspondence with the $2^k$ branch placements in $\Pi_k$.
What matters for the matter correlators, however, is only the reduced label $\bm\beta=\bm{\eta}_M\in\{0,1\}^{k-1}$.
After quotienting by this first-bit redundancy, one is left with the $2^{k-1}$ algebraically distinct ordered strings discussed in Sec.~\ref{subsec:supp_probe_expansion}.
Once those ordered strings are known, the desired braketor sectors follow by inverting the component relation in Eq.~\eqref{eq:supp_general_string_expansion}.
In this sense, the setting-dependent probe--matter assignment chooses \emph{which} correlator sample is interrogated, while the probe readout and control coefficients determine \emph{which linear combination} of the admissible orderings is observed in a given run.

The simplest example is $k=2$.
Take two probe-channel slots $p_1,p_2$ and two target matter operators $M_{m_1},M_{m_2}$.
At the raw level there are four window assignments $\eta$, matching the four possible branch placements of the two insertions.
But at the matter level there are only two reduced classes, $\bm{\eta}_M=0,1$.
Figure~\ref{fig:supp_window_pm}(b,c) shows the two settings that matter after this reduction.
In panel (b), choose the window schedule with $p_1$ active near $\tau_1$ and $p_2$ active near $\tau_2$, together with the fixed run map $m_{\eta_0}(p_1)=m_1$ and $m_{\eta_0}(p_2)=m_2$.
This realizes the class $\bm{\eta}_M=0$ and gives
\begin{equation}
\label{eq:supp_constructive_k2_strings}
\mathcal C_{(p_1,p_2),(0)}^{(2)}(\tau_1,\tau_2)
\;=\;
\expect{M_{m_1}(\tau_1)M_{m_2}(\tau_2)}.
\end{equation}
In panel (c), choose the exchanged window schedule with $p_1$ active near $\tau_2$ and $p_2$ active near $\tau_1$, together with the fixed run map $m_{\eta_1}(p_1)=m_1$ and $m_{\eta_1}(p_2)=m_2$.
This realizes the class $\bm{\eta}_M=1$; in the canonical $\widehat P_{(p_1,p_2)}$ bookkeeping it gives
\begin{equation}
\mathcal C_{(p_1,p_2),(1)}^{(2)}(\tau_1,\tau_2)
\;=\;
\expect{M_{m_1}(\tau_2)M_{m_2}(\tau_1)},
\end{equation}
where the external probe tuple remains $(p_1,p_2)$ and the exchanged time arguments record which activation window each slot occupied.
Thus the two order-$2$ strings are resolved by two different reduced timing classes, even though each class contains two raw branch assignments.
Taking symmetric and antisymmetric combinations gives the same $\widehat P_{(p_1,p_2)}$-defined braketor sectors as in Eq.~\eqref{eq:supp_k2_braketors}:
\begin{equation}
\mathcal M_{p_1,p_2}^{(0)}(\tau_1,\tau_2)
=
\frac{
\mathcal C_{(p_1,p_2),(0)}^{(2)}(\tau_1,\tau_2)
+
\mathcal C_{(p_1,p_2),(1)}^{(2)}(\tau_1,\tau_2)
}{2},
\end{equation}
\begin{equation}
\mathcal M_{p_1,p_2}^{(1)}(\tau_1,\tau_2)
=
\frac{
\mathcal C_{(p_1,p_2),(0)}^{(2)}(\tau_1,\tau_2)
-
\mathcal C_{(p_1,p_2),(1)}^{(2)}(\tau_1,\tau_2)
}{2}.
\end{equation}
This is the operational content of the control-design problem at the lowest nontrivial order:
the probe should distinguish the two admissible time-to-matter assignments, rather than merely sample the same two-point function more accurately.

At higher order the same principle persists.
One uses time-localized couplings to target the chosen ordered times and then varies the setting-dependent probe--matter assignment until the resulting data distinguish the relevant ordered strings.
That is the practical meaning of ``deconvolving'' the dynamics in the present setting.

\begin{figure*}[ht]
\centering
\includegraphics[width=1.0\textwidth]{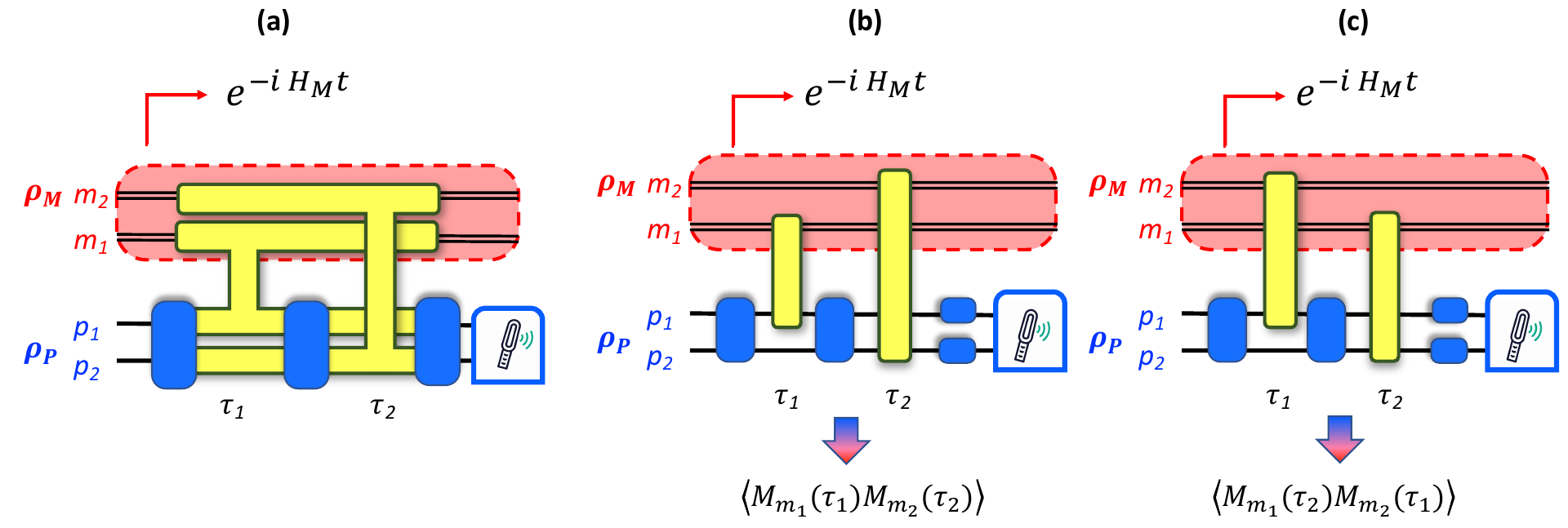}
\caption{\textbf{Always-on convolution versus windowed probe--matter control.}
(a) In an always-on or broad-window implementation, the probe signal contains a convolution over insertion times and channel placements, so the ordered matter strings must be separated by deconvolution.
(b) Window setting $\eta_0$ activates $p_1$ near $\tau_1$ and $p_2$ near $\tau_2$, with fixed run map $m_{\eta_0}(p_1)=m_1$ and $m_{\eta_0}(p_2)=m_2$; its reduced label is $\bm\eta_M=0$, and it isolates $\mathcal C_{(p_1,p_2),(0)}^{(2)}=\langle M_{m_1}(\tau_1)M_{m_2}(\tau_2)\rangle$.
(c) Window setting $\eta_1$ activates $p_1$ near $\tau_2$ and $p_2$ near $\tau_1$, with fixed run map $m_{\eta_1}(p_1)=m_1$ and $m_{\eta_1}(p_2)=m_2$; its reduced label is $\bm\eta_M=1$, and in the canonical $\widehat P_{(p_1,p_2)}$ bookkeeping it isolates $\mathcal C_{(p_1,p_2),(1)}^{(2)}=\langle M_{m_1}(\tau_2)M_{m_2}(\tau_1)\rangle$.
The desired braketor sectors are then obtained by the symmetric and antisymmetric recombinations given above.}
\label{fig:supp_window_pm}
\end{figure*}

\section{From probe correlators to entropy reconstruction}
\label{sec:supp_entropy_bridge}
The main text emphasizes that a probe does not ``measure entropy'' in one opaque step.
Rather, it measures correlators of suitable matter observables, and the entropy is subsequently reconstructed from those correlators by classical postprocessing.
This section makes that logic explicit before turning to concrete examples.

Let $A$ be a subsystem of the many-body target and let $\{R_i\}_{i\in A}$ denote a convenient operator basis on $A$.
Any quantity determined by the reduced density matrix $\rho_A$ can, in principle, be written as a functional of equal-time correlators of this basis.
Schematically,
\begin{equation}
\label{eq:supp_entropy_functional}
S(\rho_A)
=
\mathcal S\!\left[
\Big\{
\expect{R_{i_1}\cdots R_{i_m}}
\Big\}_{m\ge 1}
\right].
\end{equation}
The difficulty is therefore not conceptual but operational:
which of these correlators can be learned in the laboratory?

The answer provided by Sec.~\ref{sec:supp_general_algebra} is that ordered strings of the operators $R_i$ decompose into the nested-braketor sectors accessible to probe readout.
In the notation of Sec.~\ref{sec:supp_general_algebra}, one simply specializes the generic matter operators to the entropy basis,
\begin{equation}
M_{m(p_j)}\longrightarrow R_{i_j}.
\end{equation}
At equal time one simply takes the limit $t_1=\cdots=t_m=t$ in Eq.~\eqref{eq:supp_general_string_expansion},
\begin{equation}
\label{eq:supp_equal_time_decomp}
\expect{R_{i_1}(t)\cdots R_{i_m}(t)}
=
\sum_{\vec\mu\in\{0,1\}^{m-1}}
\mathcal M_{\vec i}^{\vec\mu}(t,\ldots,t),
\end{equation}
because the operator ordering on the left-hand side is already the reference ordering used in the definition of the nested braketors.
If one permutes the equal-time operator string, the corresponding coefficients become $\pm 1$ according to the same sign rule as Eq.~\eqref{eq:supp_general_string_expansion}.
Thus the entropy problem reduces to two steps:
\begin{enumerate}
  \item identify which equal-time correlators actually determine the entropy in the class of states under study;
  \item design probe protocols that learn the corresponding braketor sectors.
\end{enumerate}

For quasi-free matter, only second-order data are needed.
For interacting matter, higher connected correlators enter.
This is why the entropy problem fits naturally into the same operator-ordering framework as the rest of the paper.

\section{ Quasi-free entropy learning: bosons and fermions}
\label{sec:supp_entropy_gaussian}
This section expands on the Gaussian warm-up in the main text and retains the two instructive Gaussian bosonic and fermionic examples.
The message is the same in both cases:
once the relevant reduced state is Gaussian, the entropy becomes a spectral functional of low-order data on $A$.
The probe advantage then lies in learning precisely those low-order data that do not live in the retarded-response sector.

\subsection{Quasi-free bosons}
\label{subsec:supp_entropy_boson}
For bosons, the entropy-relevant operators can be chosen to be the linear quadratures themselves.
Let
\begin{equation}
R_A=(q_1,p_1,\ldots,q_{|A|},p_{|A|})^{\mathsf T}
\end{equation}
denote the vector of canonical quadratures restricted to the subsystem $A$.
In the language of Sec.~\ref{sec:supp_general_algebra}, this means that the generic matter operators are now specialised to
\begin{equation}
M_{m(p_j)}\longrightarrow R_{A,i_j}.
\end{equation}
At second order, the two equal-time bosonic sectors are simply
\begin{equation}
\mathcal M^{(0)}_{i,j}(t,t)
=
\frac12\expect{\acomm{R_{A,i}(t)}{R_{A,j}(t)}},
\qquad
\mathcal M^{(1)}_{i,j}(t,t)
=
\frac12\expect{\comm{R_{A,i}(t)}{R_{A,j}(t)}}.
\end{equation}
The second quantity is fixed by the canonical algebra,
\begin{equation}
\mathcal M^{(1)}_{i,j}(t,t)=\frac{i (\Omega_A)_{ij}}{2},
\end{equation}
so the state-dependent information sits entirely in the symmetrized sector $\mathcal M^{(0)}$.
The covariance matrix is therefore
\begin{equation}
\label{eq:supp_boson_covariance}
(V_A)_{ij}
:=
\mathcal M^{(0)}_{i,j}(t,t)-\expect{R_{A,i}(t)}\expect{R_{A,j}(t)},
\qquad
\Delta R_{A,i}:=R_{A,i}-\expect{R_{A,i}}.
\end{equation}
For a Gaussian reduced state, $V_A$ determines $\rho_{A,G}$ completely up to displacements, which do not affect the entropy.
Its symplectic eigenvalues $\{\nu_\alpha\}$ are obtained from
\begin{equation}
\label{eq:supp_boson_symplectic}
\mathrm{spec}(i\Omega_A V_A)=\{\pm \nu_\alpha\}_{\alpha=1}^{|A|},
\qquad
\Omega_A
=
\bigoplus_{\alpha=1}^{|A|}
\begin{pmatrix}
0 & 1\\
-1 & 0
\end{pmatrix}.
\end{equation}
The von Neumann entropy is then
\begin{equation}
\label{eq:supp_boson_entropy}
S(\rho_{A,G})
=
\sum_{\alpha=1}^{|A|}
\left[
(\bar n_\alpha+1)\ln(\bar n_\alpha+1)-\bar n_\alpha\ln \bar n_\alpha
\right],
\qquad
\bar n_\alpha={\nu_\alpha}-\frac12.
\end{equation}

This is the precise operational content of the bosonic example: the probe does not need to learn an entropy functional directly but to learn the equal-time anti-commutator sector $\mathcal M^{(0)}_{i,j}(t,t)$ for a basis of quadratures on $A$, together with the first moments.
The commutator sector is already known from the canonical symplectic form and carries no state-specific entropy information.
Once $V_A$ is reconstructed, the remaining steps are classical:
compute the symplectic eigenvalues of $i\Omega_A V_A$ and insert them into Eq.~\eqref{eq:supp_boson_entropy}.

\subsection{Quasi-free fermions}
\label{subsec:supp_entropy_fermion}
Correlations of the quasi-free fermions are most transparently written in terms of creation and annihilation operators on $A$.
Let
\begin{equation}
C_{ij}:=\expect{c_i^\dagger c_j},
\qquad
F_{ij}:=\expect{c_i c_j}.
\end{equation}
The pair $(C_A,F_A)$ determines a general quasi-free reduced state on $A$.
In the number-conserving case $F_A=0$, the entropy is the spectral functional
\begin{equation}
\label{eq:supp_fermion_entropy}
S(\rho_{A,G})
=
-\Tr\!\left[
C_A\ln C_A+(I-C_A)\ln(I-C_A)
\right]
=
-\sum_\alpha
\left[
\eta_\alpha\ln\eta_\alpha+(1-\eta_\alpha)\ln(1-\eta_\alpha)
\right],
\end{equation}
where $\{\eta_\alpha\}$ are the eigenvalues of $C_A$.
If pairing is present, one simply replaces $C_A$ by the corresponding Nambu covariance matrix; the entropy remains a spectral functional of the quasi-free data.

The subtlety is operational.
Equal-time anti-commutators of \emph{linear} fermionic operators are fixed by the canonical anticommutation relations and therefore carry no state-specific information.
To recover useful data, the probe must couple to \emph{bilinear} observables.
A convenient Hermitian basis is
\begin{equation}
\label{eq:supp_fermion_bilinears}
N_i:=c_i^\dagger c_i,
\qquad
X_{ij}:=c_i^\dagger c_j+c_j^\dagger c_i,
\qquad
Y_{ij}:=-i(c_i^\dagger c_j-c_j^\dagger c_i),
\end{equation}
supplemented, when pairing is relevant, by Hermitian combinations of $c_i c_j$ and $c_i^\dagger c_j^\dagger$.
Their expectation values determine the one-body data:
\begin{equation}
\expect{N_i}=C_{ii},
\qquad
\expect{X_{ij}}=2\,\mathrm{Re}\,C_{ij},
\qquad
\expect{Y_{ij}}=2\,\mathrm{Im}\,C_{ij}.
\end{equation}
Thus the quasi-free fermionic entropy problem is still a probe-learning problem, but the correct matter operators are bilinears rather than linear fields.

The reader should therefore keep the following distinction in mind.
For bosons, the entropy-relevant covariance data already live in the symmetrised sector of linear observables.
For fermions, the same idea holds, but the natural observables are bilinear because the linear equal-time anti-commutator is state-independent.
Once this is recognised, the rest of the reconstruction again follows the same pattern: learn the reduced quasi-free data on $A$, then evaluate the known entropy functional.

\section{Interacting entropy learning}
\label{sec:supp_entropy_all_order}
\subsection{Interacting entropy around a Gaussian reference}
\label{sec:supp_entropy_interacting}
The interacting problem differs from the quasi-free warm-up in only one essential way:
the reduced state $\rho_A$ is no longer determined by two-point data.
One must therefore organize the entropy around a reference state that captures the Gaussian part exactly and leaves the genuinely interacting information in higher connected correlators.

Let $\rho_{A,G}$ be the Gaussian or quasi-free reference state that matches the one- and two-point data of the true reduced state $\rho_A$ on $A$.
Define
\begin{equation}
K_{A,G}:=-\ln \rho_{A,G},
\qquad
K_A:=-\ln \rho_A,
\qquad
\delta K_A:=K_A-K_{A,G}
\end{equation}
where operators $K_{A,G}$ and $K_A$ are often called modular Hamiltonians (or the so-called entanglement Hamiltonians).
To connect the reference state and the physical state without introducing a second, competing notation, we use the single interpolation
\begin{equation}
\label{eq:supp_K_lambda}
K(\lambda)
:=
K_{A,G}+\lambda\,\delta K_A,
\qquad
0\le \lambda\le 1.
\end{equation}
The corresponding normalized state is
\begin{equation}
\label{eq:supp_rho_lambda}
\rho(\lambda)
:=
\frac{e^{-K(\lambda)}}{Z(\lambda)},
\qquad
Z(\lambda):=\Tr(e^{-K(\lambda)}),
\end{equation}
so that
\begin{equation}
\rho(0)=\rho_{A,G},
\qquad
\rho(1)=\rho_A.
\end{equation}
The entropy is therefore expanded along this single interpolation as
\begin{equation}
\label{eq:supp_entropy_series}
S(\rho_A)=S(\rho_{A,G})+\sum_{n\ge 1}\delta S^{(n)}.
\end{equation}
The point of this representation is that the Gaussian part is already handled by Sec.~\ref{sec:supp_entropy_gaussian}, while the interaction corrections are isolated in the single modular perturbation $\delta K_A$.

The practical simplification comes from matching the one- and two-point data.
Because $K_{A,G}$ is itself Gaussian or quasi-free, it is at most quadratic in the canonical operator basis on $A$.
Once the Gaussian reference has absorbed the constant, linear, and quadratic parts of the modular Hamiltonian, the leading non-Gaussian contribution to $\delta K_A$ starts at quartic order in the natural operator basis on $A$.
Using Wick or quasi-free normal ordering with respect to $\rho_{A,G}$, one writes
\begin{equation}
\label{eq:supp_deltaK_expansion}
\delta K_A
=
\sum_{m\ge m_0}
\frac{1}{m!}
\sum_{i_1,\ldots,i_m}
g^{(m)}_{i_1\cdots i_m}\;
:R_{i_1}\cdots R_{i_m}:_G,
\end{equation}
with $m_0=4$ in the examples below.
Equation~\eqref{eq:supp_deltaK_expansion} makes the logic of the main text precise:
interacting entropy is governed by higher connected correlators of local operators on $A$, and these are exactly the many-body objects that the probe-learning framework is built to access.

\subsection{General all-order formula}
\label{subsec:supp_all_order_general}

The following sections continue from the setup above.
No new modular Hamiltonian is introduced:
the only perturbation parameter is the bookkeeping variable $\lambda$ in Eq.~\eqref{eq:supp_K_lambda}, and the physical target state is recovered at $\lambda=1$.
We first derive the general all-order formula and then extract the leading quartic correction relevant for the present paper.

Define modular-time evolution with respect to the Gaussian reference,
\begin{equation}
\mathcal O(\tau):=e^{\tau K_{A,G}}\,\mathcal O\,e^{-\tau K_{A,G}},
\qquad
0\le \tau\le 1,
\end{equation}
and let $T_\tau$ denote ordering in $\tau$.
Because $e^{-K_{A,G}}=\rho_{A,G}$ and $\Tr(\rho_{A,G})=1$, the Duhamel factorization gives
\begin{equation}
\label{eq:supp_duhamel_exact}
e^{-K(\lambda)}
=
\rho_{A,G}\,
T_\tau
\exp\!\left[
-\lambda\int_0^1 d\tau\,\delta K_A(\tau)
\right].
\end{equation}
Hence
\begin{equation}
\label{eq:supp_Z_ratio}
Z(\lambda)
=
\Big\langle
T_\tau
\exp\!\left[
-\lambda\int_0^1 d\tau\,\delta K_A(\tau)
\right]
\Big\rangle_G,
\qquad
\langle \cdot\rangle_G:=\Tr(\rho_{A,G}\,\cdot),
\end{equation}
with $Z(0)=1$.
The linked-cluster expansion then yields
\begin{equation}
\label{eq:supp_logZ_cumulant}
\ln Z(\lambda)
=
\sum_{n=1}^\infty
\frac{(-\lambda)^n}{n!}\,
C_n[\delta K_A],
\end{equation}
with
\begin{equation}
\label{eq:supp_Cn_def}
C_n[\delta K_A]
:=
\int_0^1 d\tau_1\cdots d\tau_n\;
\Big\langle
T_\tau \delta K_A(\tau_1)\cdots \delta K_A(\tau_n)
\Big\rangle_{G,c}.
\end{equation}
Using $\partial_\lambda \ln Z(\lambda)=-\langle \delta K_A\rangle_\lambda$, the entropy may be rewritten as
\begin{equation}
\label{eq:supp_entropy_rewrite}
S(\lambda)
:=
-\Tr\!\big[\rho(\lambda)\ln \rho(\lambda)\big]
=
\Tr\!\big[\rho(\lambda)K_{A,G}\big]
+\ln Z(\lambda)
-\lambda\frac{d}{d\lambda}\ln Z(\lambda).
\end{equation}
Expanding the normalized expectation of $K_{A,G}$ gives
\begin{equation}
\label{eq:supp_CnK_def}
C_n[K_{A,G};\delta K_A]
:=
\int_0^1 d\tau_1\cdots d\tau_n\;
\Big\langle
T_\tau K_{A,G}\,\delta K_A(\tau_1)\cdots \delta K_A(\tau_n)
\Big\rangle_{G,c},
\end{equation}
through
\begin{equation}
\label{eq:supp_KAG_exp}
\Tr\!\big[\rho(\lambda)K_{A,G}\big]
=
S(\rho_{A,G})
+
\sum_{n=1}^{\infty}
\frac{(-\lambda)^n}{n!}\,
C_n[K_{A,G};\delta K_A].
\end{equation}
Combining Eqs.~\eqref{eq:supp_logZ_cumulant}, \eqref{eq:supp_entropy_rewrite}, and \eqref{eq:supp_KAG_exp}, and then evaluating the series at $\lambda=1$, gives the all-order entropy expansion
\begin{equation}
\label{eq:supp_all_order_entropy}
S(\rho_A)=S(\rho_{A,G})+\sum_{n=1}^\infty \delta S^{(n)},
\qquad
\delta S^{(n)}
=
\frac{(-1)^n}{n!}
\Big[
(1-n)\,C_n[\delta K_A]+C_n[K_{A,G};\delta K_A]
\Big].
\end{equation}

Equation~\eqref{eq:supp_all_order_entropy} is the compact formal statement underlying the interacting discussion in the main text.
It says that the $n$th entropy correction is a connected modular $n$-point correlator of the single perturbation $\delta K_A$ around the single reference state $\rho_{A,G}$.
Once Eq.~\eqref{eq:supp_deltaK_expansion} is substituted into Eq.~\eqref{eq:supp_all_order_entropy}, the operator content of each correction becomes explicit.

\subsection{Low-order consequences and leading quartic correction}
\label{subsec:supp_all_order_operator_content}
The all-order formula immediately reproduces the familiar low-order expansion.
Linearizing Eq.~\eqref{eq:supp_rho_lambda} around $\lambda=0$ and then setting $\lambda=1$ gives
\begin{equation}
\label{eq:supp_duhamel_linear}
\delta \rho_A
:=
\rho_A-\rho_{A,G}
=
\;-\int_0^1 d\tau\;
\rho_{A,G}^{\,1-\tau}\,\delta K_A\,\rho_{A,G}^{\,\tau}
\;+\;
\rho_{A,G}\,\Tr(\rho_{A,G}\delta K_A)
\;+\;
\mathcal O(\delta K_A^2).
\end{equation}
The first two corrections may be written as
\begin{equation}
\label{eq:supp_entropy_first_second}
\delta S^{(1)}
=
\Tr(\delta\rho_A\,K_{A,G}),
\qquad
\delta S^{(2)}
=
-\frac12\Tr\!\left[
\delta\rho_A\,J_{\rho_{A,G}}^{-1}(\delta\rho_A)
\right],
\end{equation}
where
\begin{equation}
J_{\rho}^{-1}(X)
:=
\int_0^\infty ds\;
(\rho+s)^{-1}X(\rho+s)^{-1}.
\end{equation}
Equation~\eqref{eq:supp_entropy_first_second} is the Bogoliubov--Kubo--Mori quadratic form appropriate to the entropy expansion.

The Gaussian matching now becomes crucial.
Since $K_{A,G}$ is quadratic or quasi-free in the canonical operators on $A$, while $\rho_{A,G}$ is chosen to match the one- and two-point data of $\rho_A$, one has
\begin{equation}
\delta S^{(1)}
=
\Tr\!\big[(\rho_A-\rho_{A,G})K_{A,G}\big]
=0.
\end{equation}
Thus, for the quartic-leading perturbations considered below, the first nontrivial interacting entropy correction starts at second order in the modular couplings.

If the leading perturbation in Eq.~\eqref{eq:supp_deltaK_expansion} is quartic, define
\begin{equation}
O_{ijkl}:=:R_iR_jR_kR_l:_G,
\qquad
\delta K_A
=
\frac{1}{4!}\sum_{ijkl} g_{ijkl}\,O_{ijkl}
\;+\;\mathcal O(R^6).
\end{equation}
Because the operators $O_{ijkl}$ are normal ordered with respect to $\rho_{A,G}$, one has $\expect{O_{ijkl}}_{\rho_{A,G}}=0$, so the normalization term in Eq.~\eqref{eq:supp_duhamel_linear} vanishes at leading quartic order.
Then Eq.~\eqref{eq:supp_all_order_entropy} implies that the leading non-Gaussian entropy correction is the quadratic term
\begin{equation}
\label{eq:supp_leading_quartic_entropy}
\delta S_{\mathrm{lead}}
=
\delta S^{(2)}
=
-\frac{1}{2(4!)^2}
\sum_{ijkl,mnpq}
g_{ijkl}\,
\mathcal K^{(G)}_{ijkl,mnpq}\,
g_{mnpq}
\;+\;
\mathcal O(g^3),
\end{equation}
with Gaussian modular kernel
\begin{equation}
\label{eq:supp_leading_quartic_kernel}
\mathcal K^{(G)}_{ijkl,mnpq}
:=
\int_0^1 d\tau\;
\Big\langle
T_\tau\,O_{ijkl}(\tau)\,O_{mnpq}(0)
\Big\rangle_{G,c}.
\end{equation}
Equivalently, one may trade the unknown modular couplings $g_{ijkl}$ for the measurable quartic cumulants
\begin{equation}
\kappa_{ijkl}
:=
\expect{O_{ijkl}}_{\rho_A},
\end{equation}
obtaining
\begin{equation}
\label{eq:supp_leading_quartic_kappa}
\delta S_{\mathrm{lead}}
=
-\frac12
\sum_{ijkl,mnpq}
\kappa_{ijkl}\,
\widetilde{\mathcal K}^{(G)}_{ijkl,mnpq}\,
\kappa_{mnpq}
\;+\;
\mathcal O(\kappa^3),
\end{equation}
for a kernel $\widetilde{\mathcal K}^{(G)}$ fixed entirely by the Gaussian reference.

This is the precise leading-order bridge to the main text.
The Gaussian reference fixes the modular-time kernels.
What must be learned experimentally are the non-Gaussian quartic cumulants on $A$, and those are many-body correlator data of exactly the type discussed in Sec.~\ref{sec:supp_general_algebra}.

\subsection{Example I: weakly interacting bosons}
\label{subsec:supp_entropy_phi4}
Consider a weakly interacting bosonic field or lattice system with Hamiltonian
\begin{equation}
H
=
H_0+\frac{\lambda}{4!}\int d^dx\,:\phi(x)^4: ,
\end{equation}
where $H_0$ is quadratic.
Let $A$ be the subsystem of interest and let $\rho_{A,G}$ be the Gaussian reference state matching the one- and two-point data of $\rho_A$ on $A$.
The leading modular correction is the continuum specialization of the quartic ansatz above,
\begin{equation}
\label{eq:supp_phi4_deltaK}
\delta K_A
\simeq
\frac{\lambda_{\mathrm{eff}}}{4!}
\int_A d^dx\;
f_A(x)\,:\phi(x)^4:_G,
\end{equation}
where $f_A(x)$ is the geometric weight associated with the subsystem and $\lambda_{\mathrm{eff}}$ is the effective interaction strength at the scale of $A$.

Define the local quartic cumulant field
\begin{equation}
\kappa(x):=\expect{:\phi(x)^4:_G}_{\rho_A}.
\end{equation}
For the Gaussian reference, Wick's theorem gives the continuum version of Eq.~\eqref{eq:supp_leading_quartic_kernel},
\begin{equation}
\label{eq:supp_phi4_wick}
\Big\langle
:\phi(x)^4:_G(\tau)\,
:\phi(y)^4:_G
\Big\rangle_{\rho_{A,G}}
=
4!\,
\big(G_G(\tau;x,y)\big)^4,
\end{equation}
where
\begin{equation}
G_G(\tau;x,y)
:=
\expect{\phi(\tau,x)\phi(0,y)}_{\rho_{A,G}}.
\end{equation}
The Gaussian kernel relevant to the entropy correction is therefore
\begin{equation}
\label{eq:supp_phi4_kernel}
\mathcal K_{\phi}^{(G)}(x,y)
:=
\int_0^1 d\tau\;
4!\,
\big(G_G(\tau;x,y)\big)^4.
\end{equation}
At leading nontrivial order one obtains
\begin{equation}
\label{eq:supp_phi4_entropy}
\delta S_A^{(2)}
=
-\frac{\lambda_{\mathrm{eff}}^2}{2(4!)^2}
\int_A d^dx\int_A d^dy\;
f_A(x)\,\mathcal K_{\phi}^{(G)}(x,y)\,f_A(y)
\;+\;
\mathcal O(\lambda_{\mathrm{eff}}^3).
\end{equation}
Equivalently, after eliminating the unknown coupling profile in favor of the measurable cumulant field $\kappa(x)$,
\begin{equation}
\delta S_A^{(2)}
=
-\frac12
\int_A d^dx\int_A d^dy\;
\kappa(x)\,
\widetilde{\mathcal K}^{(G)}(x,y)\,
\kappa(y)
\;+\;
\mathcal O(\kappa^3),
\end{equation}
for a kernel $\widetilde{\mathcal K}^{(G)}$ determined entirely by the Gaussian reference.

The example is meant to be read operationally.
The Gaussian part of the entropy comes from the covariance matrix, as in Sec.~\ref{sec:supp_entropy_gaussian}.
The leading interaction correction comes from the quartic cumulant field $\kappa(x)$.
Thus the probe task is to learn both the two-point Gaussian data and the quartic fluctuation data on $A$.

\subsection{Example II: weakly interacting fermions}
\label{subsec:supp_entropy_hubbard}
As a fermionic counterpart, consider the Hubbard model
\begin{equation}
H
=
-\sum_{\langle ij\rangle,\sigma} t_{ij}\,c_{i\sigma}^\dagger c_{j\sigma}
+
U\sum_i n_{i\uparrow}n_{i\downarrow},
\qquad
n_{i\sigma}:=c_{i\sigma}^\dagger c_{i\sigma}.
\end{equation}
For weak interaction strength it is natural to choose $\rho_{A,G}$ as the quasi-free reference state matching the one- and two-point fermionic data on $A$.
The leading local non-Gaussian operator is then the double-occupancy fluctuation, which is the lattice specialization of the quartic operator basis used above,
\begin{equation}
\label{eq:supp_hubbard_Oi}
O_i
:=
:n_{i\uparrow}n_{i\downarrow}:_G
=
n_{i\uparrow}n_{i\downarrow}
-\expect{n_{i\uparrow}}_G\expect{n_{i\downarrow}}_G.
\end{equation}
The modular perturbation is modeled as
\begin{equation}
\label{eq:supp_hubbard_deltaK}
\delta K_A
\simeq
U_{\mathrm{eff}}
\sum_{i\in A}
f_A(i)\,O_i .
\end{equation}
Define the corresponding cumulant variables
\begin{equation}
\kappa_i:=\expect{O_i}_{\rho_A}
=
\expect{n_{i\uparrow}n_{i\downarrow}}_{\rho_A}
-\expect{n_{i\uparrow}}_{\rho_A}\expect{n_{i\downarrow}}_{\rho_A}.
\end{equation}
The Gaussian kernel is the discrete analogue of Eq.~\eqref{eq:supp_leading_quartic_kernel},
\begin{equation}
\label{eq:supp_hubbard_kernel}
\mathcal K_{ij}^{(G)}
:=
\int_0^1 d\tau\;
\expect{O_i(\tau)\,O_j}_{\rho_{A,G}}.
\end{equation}
For the quasi-free reference, Wick factorization gives
\begin{equation}
\label{eq:supp_hubbard_wick}
\expect{O_i(\tau)\,O_j}_{\rho_{A,G}}
=
|G_{\uparrow}(\tau;i,j)|^2\,
|G_{\downarrow}(\tau;i,j)|^2,
\end{equation}
where
\begin{equation}
G_{\sigma}(\tau;i,j)
:=
\expect{c_{i\sigma}^\dagger(\tau)c_{j\sigma}(0)}_{\rho_{A,G}}.
\end{equation}
Therefore the leading entropy correction reads
\begin{equation}
\label{eq:supp_hubbard_entropy}
\delta S_A^{(2)}
=
-\frac{U_{\mathrm{eff}}^2}{2}
\sum_{i,j\in A}
f_A(i)\,\mathcal K_{ij}^{(G)}\,f_A(j)
\;+\;
\mathcal O(U_{\mathrm{eff}}^3).
\end{equation}
As in the bosonic example, one can equivalently rewrite the result in terms of the measurable connected cumulants $\kappa_i$ and a Gaussian kernel determined by the reference state.

This example clarifies the fermionic story for the reader.
The quasi-free entropy of Sec.~\ref{sec:supp_entropy_gaussian} depends only on one-body data.
Interactions enter through local double-occupancy fluctuations, which are quartic operators.
The same probe-learning logic therefore survives, but the relevant observables move up the correlator hierarchy.

\stopcontents[supplement]
\end{document}